\documentclass[
aip,jmp,
reprint,
a4paper,
longbibliography,
nofootinbib,
nobibnotes,
]{revtex4-2}
\usepackage[utf8]{inputenc}
\usepackage[T1]{fontenc}

\usepackage{amssymb}
\usepackage{amsthm}
\usepackage{braket}
\usepackage{graphicx}
\usepackage{mathtools}
\usepackage{array}
\newcolumntype{C}{>{$}c<{$}}
\usepackage{hyperref}
\hypersetup{citecolor=magenta}
\hypersetup{colorlinks=true}
\hypersetup{linkcolor=blue}
\hypersetup{urlcolor=blue}

\newcommand\id\openone
\newcommand\oo{\mathbb O}
\newcommand\ketbra[2]{\ket{#1}\!\bra{#2}}

\newcommand\naturals{\mathbb N}
\newcommand\complexes{\mathbb C}
\newcommand\chif{\chi_\mathrm{f}}
\renewcommand\vec\mathbf

\DeclareMathOperator\argmin{arg~min}
\DeclareMathOperator\diag{diag}
\DeclareMathOperator\rank{rank}
\DeclareMathOperator\tr{tr}
\DeclareMathOperator\rankb{RANK}
\DeclareMathOperator\lrankb{LRANK}
\DeclareMathOperator\stabb{STAB}
\DeclareMathOperator\thb{TH}

\DeclarePairedDelimiter\abs\lvert\rvert
\DeclarePairedDelimiter\ceil\lceil\rceil
\DeclarePairedDelimiter\norm\lVert\rVert

\newtheorem{theorem}{Theorem}
\newtheorem{assertion}[theorem]{Assertion}
\newtheorem{lemma}[theorem]{Lemma}
\newtheorem{corollary}[theorem]{Corollary}

\usepackage{xcolor}
\definecolor{darkred}{rgb}{0.7 0.0 0.0}

\begin{document}
\title{State-independent quantum contextuality with projectors of nonunit rank}

\author{Zhen-Peng Xu}
\email{zhen-peng.xu@uni-siegen.de}
\author{Xiao-Dong Yu}
\email{xiao-dong.yu@uni-siegen.de}
\affiliation{Naturwissenschaftlich--Technische Fakult\"{a}t, Universit\"{a}t Siegen, 
Walter-Flex-Stra{\ss}e 3, 57068 Siegen, Germany}
\author{Matthias Kleinmann}
\email{matthias.kleinmann@uni-siegen.de}
\affiliation{Faculty of Physics, University of Duisburg--Essen, Lotharstra{\ss}e  
1, 47048 Duisburg, Germany}
\affiliation{Naturwissenschaftlich--Technische Fakult\"{a}t, Universit\"{a}t Siegen, 
Walter-Flex-Stra{\ss}e 3, 57068 Siegen, Germany}

\begin{abstract}
Virtually all of the analysis of quantum contextuality is restricted to the 
case where events are represented by rank-one projectors.
This restriction is arbitrary and not motivated by physical considerations.
We show here that loosening the rank constraint opens a new realm of quantum 
contextuality and we demonstrate that state-independent contextuality can even 
require projectors of nonunit rank.
This enables the possibility of state-independent contextuality with less than 
  13 projectors, which is the established minimum for the case of rank one.
We prove that for any rank, at least 9 projectors are required.
Furthermore, in an exhaustive numerical search we find that 13 projectors are 
also minimal for the cases where all projectors are uniformly of rank two or 
uniformly of rank three.
\end{abstract}

\maketitle

\section{Introduction}

Experiments provide strong evidence that the measurements on quantum systems 
cannot be reproduced by any noncontextual hidden variable model (NCHV).
In a NCHV model each outcome of any measurement has a preassigned value and 
this value in particular does not depend on which other properties are obtained 
alongside.
This phenomenon is called quantum contextuality.
Being closely connected to the incompatibility of observables 
\cite{xu2019necessary}, quantum contextuality is the underlying feature of 
quantum theory that enables, for example, the violation of Bell inequalities 
\cite{bell1964epr}, enhanced quantum communication \cite{cubitt2010improving, 
saha2019state}, cryptographic protocols \cite{cabello2011hybrid, 
ekert1991quantum}, quantum enhanced computation \cite{howard2014contextuality, 
raussendorf2013contextuality}, and quantum key distribution 
\cite{barrett2005no}. 

The first example of quantum contextuality was found by Kochen and Specker 
\cite{kochen1968problem} and requires 117 rank-one projectors.
Subsequently the number of projectors was reduced until it was proved that the 
minimal set has 18 rank-one projectors \cite{cabello1996bell}. This analysis 
was based on the particular type of contradiction between value assignments and 
projectors that was already used in the original proof by Kochen and Specker. 
The situation changed with the introduction of state-independent 
noncontextuality inequalities, where any NCHV model obeys the inequality, while 
it is violated for any quantum state and a certain set of projectors. With this 
enhanced definition of state-independent contextuality (SIC), Yu and Oh 
\cite{yu2012stateindependent} found an instance of SIC with only 13 rank-one 
projectors and subsequently it was proved that this set is minimal 
\cite{cabello2016quantum} provided that all projectors are of rank one.
Note that the iconic example of the Peres--Mermin square 
\cite{peres1990incompatible,mermin1990simple} uses 9 observables with two-fold 
degenerate eigenspaces, but they are combined to 6 measurements of 24 rank-one 
projectors.

In contrast, SIC involving nonunit rank projectors has been rarely considered. 
To the best of our knowledge, the only examples \cite{kernaghan1995kochen, 
mermin1990simple, toh2013kochen, toh2013stateindependent} which use nonunit 
rank are based on the Mermin star \cite{mermin1993hidden}.
In these examples it was shown that nonunit projectors are sufficient for SIC, 
but it was not shown whether nonunit projectors are also necessary for SIC.
Furthermore, in a graph theoretical analysis by Ramanathan and Horodecki 
\cite{ramanathan2014necessary} a necessary condition for SIC was provided which 
also allows one to study the case of nonunit rank.

In this article, we develop mathematical tools to analyze SIC for the case of 
 nonunit rank.
We first show that in certain situations nonunit rank is necessary for SIC.
Then we approach the question whether projectors with nonunit rank enable SIC 
with less than 13 projectors.
We find that in this case at least 9 projectors are required.
For the special cases of SIC where all projectors are of rank 2 or rank 3 we 
find strong numerical evidence that 13 is indeed the minimal number of 
projectors.

This paper is structured as follows.
In Section~\ref{sec:graph} we give an introduction to quantum contextuality 
using the graph theoretic approach.
We extend this discussion to SIC in Section~\ref{sec:sic} and we give an 
example where rank-two projectors are necessary for SIC.
In Section~\ref{sec:rank} we provide a general analysis of the case of nonunit 
rank and show that scenarios with 8 or less projectors do not feature SIC, 
irrespective of the involved ranks.
This analysis is used in Section~\ref{sec:exhaustive} to show in an exhaustive 
numerical search that all graphs smaller than the graph given by Yu and Oh do 
not have SIC, if the rank of all projectors is 2 or 3.
We conclude in Section~\ref{sec:conclusion} with a discussion of our results.

\section{Contextuality and the graph theoretic approach}\label{sec:graph}

Our analysis is based on the graph theoretic approach to quantum contextuality 
\cite{cabello2014graphtheoretic}. In this approach an exclusivity graph $G$ 
with vertices $V(G)$ and edges $E(G)$ specifies the exclusivity relations in a 
contextuality scenario.
The vertices represent events and two events are exclusive if they are 
connected by an edge.
The cliques of the graph form the contexts of the scenario.
(In Appendix~\ref{app:basicgt} we give definitions of essential terms from 
graph theory.)
Recall that an event is a class of outcomes in an experiment and two events are 
exclusive if they cannot be obtained simultaneously in any experiment.
We consider now two types of models implementing the exclusivity graph, quantum 
models and noncontextual hidden variable models.

In a quantum model of the exclusivity graph $G$ one assigns projectors $\Pi_k$ 
to each event $k$ such that $\sum_{k\in C}\Pi_k$ is again a projector for every 
context $C$.
This is equivalent to having $\Pi_k\Pi_l=0$ for any two exclusive events $k$ 
and $l$.
With such an assignment and a quantum state $\rho$ one obtains the probability 
for the event $k$ as
\begin{equation}
 P_\mathrm{QT}(k)= \tr(\rho \Pi_k).
\end{equation}
The set of all probability assignments $P_\mathrm{QT}$ that can be reached with 
some projectors $(\Pi_k)_k$ and some state $\rho$ is a convex set which 
coincides\cite{cabello2014graphtheoretic} with the theta body $\thb(G)$ of the 
graph $G$.

In contrast, in a NCHV model for the exclusivity graph $G$ the events are 
predetermined by a hidden variable $\lambda\in \Lambda$.
That is, to each event $k$ one associates a response function $R_k\colon 
\Lambda\to \set{0,1}$.
For a context $C$ the function $\lambda\mapsto \sum_{k\in C}R_k(\lambda)$ has 
to be again a response function, which is equivalent to 
$R_k(\lambda)R_l(\lambda)=0$ for all $\lambda$ and any pair of exclusive events 
$k$ and $l$.
The probability of an event $k$ is now given by
\begin{equation}
 P_\mathrm{NCHV}(k)= \sum_{\lambda\in \Lambda} \mu(\lambda)R_k(\lambda),
\end{equation}
where $\mu$ is some probability distribution over the hidden variable space 
$\Lambda$.
The set of all probability assignments $P_\mathrm{NCHV}$ that can be reached 
with some response functions $(R_k)_k$ and some distribution $\mu$ forms a 
polytope which can be shown \cite{cabello2014graphtheoretic} to be the stable 
set $\stabb(G)$ of the graph $G$.

Quantum models and NCHV models are both noncontextual in the sense that the 
computation of the probability $P(k)$ of an event $k$ does not depend on the 
context in which $k$ is contained.
Quantum contextuality occurs now for an exclusivity graph $G$ if we can find a 
quantum model with probability assignment $P_\mathrm{QT}$ which cannot be 
achieved by any NCHV model and hence 
$P_\mathrm{QT}\in\thb(G)\setminus\stabb(G)$.
Since $\stabb(G)$ is convex, it is possible to find nonnegative numbers 
$(w_k)_{k\in V(G)}\equiv \vec w$ such that
\begin{equation}
 I_{\vec w}\colon P\mapsto \sum_k w_k P(k)
\end{equation}
separates all NCHV models from some of the quantum models.
That is, there exists some $\alpha$, such that $I_{\vec w}(P_\mathrm{NCHV})\le 
\alpha$ holds for any $P_\mathrm{NCHV}\in \stabb(G)$, while $I_{\vec 
w}(P_\mathrm{QT}) > \alpha$ holds true for some $P_\mathrm{QT}\in \thb(G)$.
This can be further formalized by realizing that the weighted independence 
number \cite{grotschel1984polynomial} $\alpha(G,\vec w)$ is exactly the maximal 
value that $I_{\vec w}$ attains within $\stabb(G)$ and similarly that the 
weighted Lovász number \cite{lovasz1979shannon} $\vartheta(G,\vec w)$ is 
exactly the maximum of $I_{\vec w}$ over $\thb(G)$.
Consequently the inequality $I_{\vec w}(P_\mathrm{NCHV})\le \alpha(G,\vec w)$ 
holds for all NCHV probability assignments and this inequality is violated by 
some quantum probability assignment if and only if 
\cite{cabello2014graphtheoretic} $\vartheta(G,\vec w)>\alpha(G,\vec w)$ holds.
In addition, one can show \cite{cabello2014graphtheoretic} that the value of 
$\vartheta(G,\vec w)$ can always be attained for some quantum model employing 
only rank-one projectors.

\section{State-independent contextuality and nonunit rank}\label{sec:sic}

\begin{figure}
\centering
\includegraphics[width=0.8\textwidth]{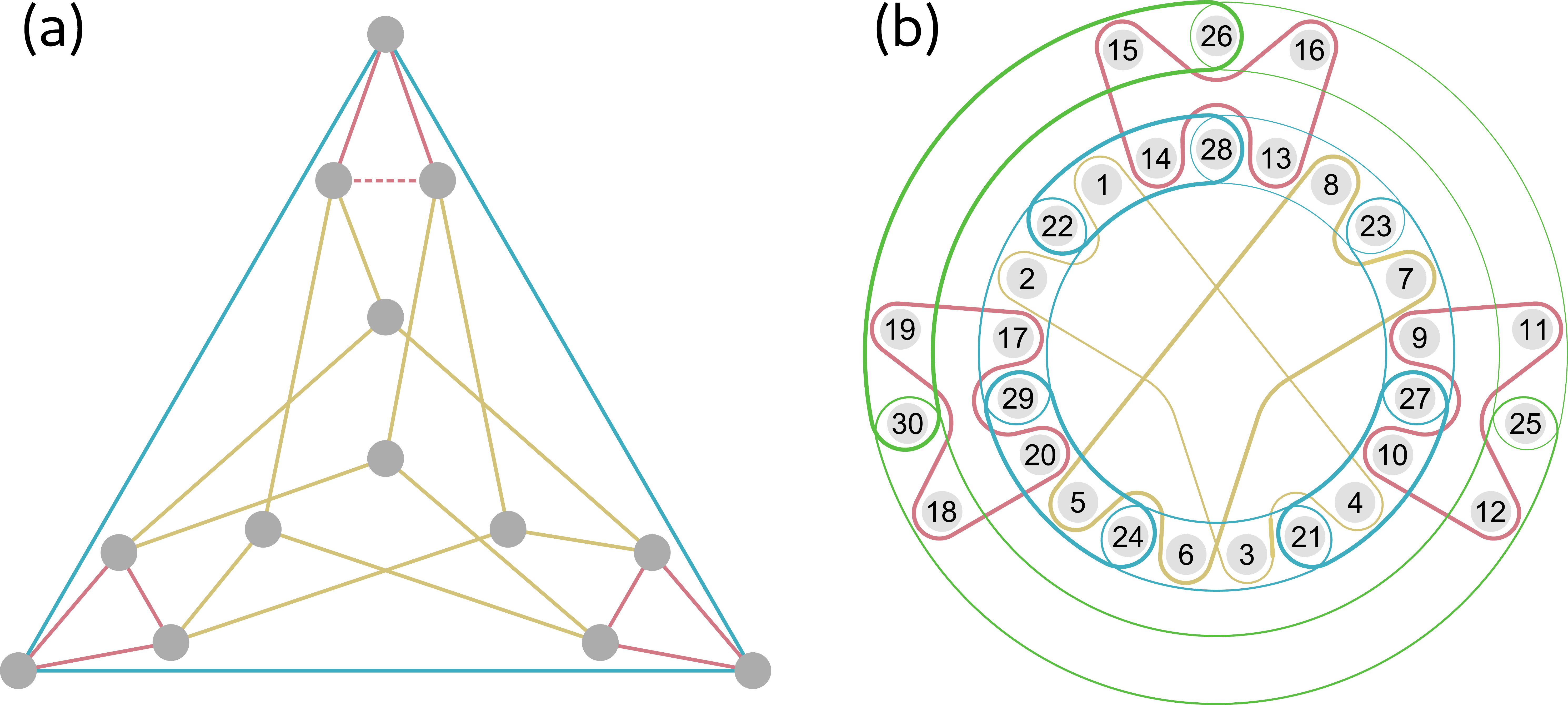}
\caption{\label{fig:g30}%
(a)~The graph $G_\mathrm{YO}$ with 13 vertices and 24 edges.
Removing the dashed edge yields the graph $G'_\mathrm{YO}$. (b)~Representation 
of the graph $G_\mathrm{Toh}$. Any subset of the vertices $1,2,\dotsc,30$ 
enclosed by a line forms a clique, that is, all vertices in any of the sets
$\set{1,2,3,4}$,
$\set{5,6,7,8}$,
$\set{9,10,11,12}$,
$\set{13,14,15,16}$,
$\set{17,18,19,20}$,
$\set{28,14,1,22}$,
$\set{22,2,17,29}$,
$\set{29,20,5,24}$,
$\set{24,6,3,21}$,
$\set{21,4,10,27}$,
$\set{27,9,7,23}$,
$\set{23,8,13,28}$,
$\set{26,15,19,30}$,
$\set{30,18,12,25}$, or
$\set{25,11,16,26}$
are mutually connected by an edge.}
\end{figure}

The discussion so far concerns quantum models as being specified by the 
projectors assigned to each event together with a quantum state.
In SIC one removes the quantum state from the specification of a quantum model 
and instead requires that probabilities from the quantum model cannot be 
reproduced by a NCHV model, independent of the quantum state.
Therefore we consider the set of probability assignments formed by all quantum 
states and fixed projectors $(\Pi_k)_k$,
\begin{equation}
 \mathcal P_\mathrm{SIC}=\set{ P \colon k\mapsto \tr(\rho\Pi_k) | \rho
 \text{ is a quantum state} }.
\end{equation}
This set is also convex, since $P$ is linear and the set of quantum states is 
convex.
Hence, in the case of SIC it is again possible to find nonnegative numbers 
$(w_k)_k\equiv \vec w$ such that $I_{\vec w}$ separates $\stabb(G)$ from 
$\mathcal P_\mathrm{SIC}$.
Therefore, it holds that
\begin{equation}
 \sum_k w_k \tr(\rho \Pi_k) > \alpha(G,\vec w), \text{ for all } \rho,
\end{equation}
or, equivalently, that the eigenvalues of
\begin{equation}
 \sum_k w_k \Pi_k -\alpha(G,\vec w)
\end{equation}
are all strictly positive.

We say that the projectors $(\Pi_k)_k$ of a quantum model of $G$ form a 
rank-$\vec r$ projective representation\footnote{%
An projective representation obeys $\Pi_k\Pi_l=0$ if $[k,l]\in E(G)$.
In contrast, an orthogonal representation obeys $\braket{\psi_k|\psi_l}=0$ if 
$[k,l]\in E(\overline G)$.%
} of $G$, when $\vec r=(r_k)_{k\in V(G)}$ with $r_k$ the rank of $\Pi_k$.
The smallest known contextuality scenario which allows SIC is given by the 
exclusivity graph $G_\mathrm{YO}$ with 13 vertices 
\cite{yu2012stateindependent}. This graph is shown in Figure~\ref{fig:g30}~(a).
For this scenario it is sufficient to consider rank-one projective 
representations.
It also has been shown that no exclusivity graph with 12 or less vertices 
allows SIC \cite{cabello2016quantum}, provided that all projectors are of rank 
one, $\vec r=\vec 1$.
But this does not yet show that SIC requires 13 projectors, since it is 
possible that a contextuality scenario features SIC only if some of the 
projectors are of nonunit rank.

This rises the question whether projectors of nonunit rank can be of advantage 
regarding SIC.
We now show that this is the case by analyzing the exclusivity graph 
$G_\mathrm{Toh}$ with 30 vertices \cite{toh2013stateindependent}.
This graph is shown in Figure~\ref{fig:g30}~(b).
One can find a rank-two projective representation of this graph 
\cite{toh2013stateindependent}, such that $\sum_k \Pi_k = 7+\frac12$.
Since the independence number of $G_\mathrm{Toh}$ is $7$, that is, 
$\alpha(G_\mathrm{Toh})\equiv \alpha(G_\mathrm{Toh}, \vec 1) = 7$, this shows 
that rank two is sufficient for SIC in this scenario.

For necessity, we show that no rank-one projective representation featuring SIC 
of $G_\mathrm{Toh}$ exists.
We first note that such a representation would be necessarily constructed in a 
four-dimensional Hilbert space.
This is the case because the largest clique of $G_\mathrm{Toh}$ has size four 
and hence any projective representation must contain at least four mutually 
orthogonal projectors of rank one.
For an upper bound on the dimension $d$ of any projective representation 
featuring SIC we use the result \cite{ramanathan2014necessary, 
cabello2015necessary}
\begin{equation}\label{e:ramanathan}
 d < \chif(G),
\end{equation}
where $\chif(G)$ denotes the fractional chromatic number of $G$.
One finds $\chif(G_\mathrm{Toh})= 4+\frac27$ implying $d\le 4$.
We do not find any rank-one projective representation of $G_\mathrm{Toh}$ in 
dimension $d=4$ using the numerical methods discussed in 
Section~\ref{sec:seesaw} and in Appendix~\ref{app:gtoh} we prove also 
analytically that no such representation exists.

\section{Graph approach for projective representations of arbitrary 
rank}\label{sec:rank}

The example of the previous section showed that considering projective 
representations of nonunit rank can be necessary for the existence of a quantum 
model with SIC.
Since the case of rank-one has already been analyzed in detail, it is helpful 
to reduce the case of nonunit rank to the case of rank one.
To this end we adapt the notation \cite{schrijver2004fractional} $G^{\vec r}$ 
for the graph where each vertex $k$ is replaced by a clique $C_k$ of size $r_k$ 
and all vertices between two cliques $C_k$ and $C_\ell$ are connected when 
$[k,\ell]$ is an edge.
See Figure~\ref{fig:gsp} for an illustration.
That is,
\begin{align}
 V(G^{\vec r})&=\set{ (k,i) | k\in V(G),\; i=1,2,\dotsc, r_k}, \\
 E(G^{\vec r})&=\set{[(k,i),(\ell,j)]| [k,\ell]\in E(G) \text{ or } (k=\ell 
 \text{ and } i\ne j)}.
\end{align}
\begin{figure}
\centering
\includegraphics[width=0.5\textwidth]{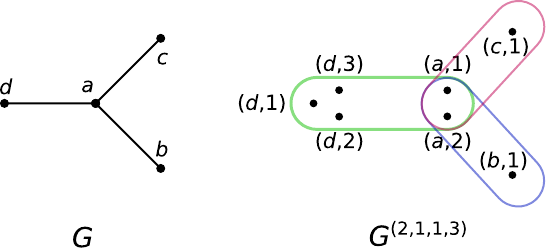}
\caption{\label{fig:gsp}%
Illustration of the graph $G^{\vec r}$.
$G$ has vertices $a,b,c,d$ and here $r_a=2$, $r_b=1$, $r_c=1$, $r_d=3$.
In the product graph, vertices enclosed by a line form a clique, that is, they 
are all mutually connected by an edge.
}
\end{figure}
The construction of $G^{\vec r}$ is such that if $(\Pi_{k,i})_{k,i}$ is a 
rank-one projective representation of $G^{\vec r}$, then evidently $\Pi_k= 
\sum_i \Pi_{k,i}$ defines a rank-$\vec r$ projective representation of $G$.
Vice versa, if $(\Pi_k)_k$ is a rank-$\vec r$ projective representation of $G$, 
then one can immediately construct a rank-one projective representation of 
$G^{\vec r}$ by decomposing each projector $\Pi_k$ into rank-one projectors 
$(\Pi_{k,i})_i$ such that $\Pi_k=\sum_i \Pi_{k,i}$.

For a given graph $G$ we denote by $d_\pi(G,\vec r)$ the minimal dimension 
which admits a rank-$\vec r$ projective representation and by $\chif(G,\vec 
r)$ the fractional chromatic number for the graph $G$ with vertex weights $\vec 
r\in \naturals^{\abs{V(G)}}$.
In addition we abbreviate the Lovász function of the complement graph by 
$\overline\vartheta(G,\vec r)= \vartheta(\overline G,\vec r)$.
For these three functions we omit the second argument if $r_k=1$ for all $k$, 
that is, $\chif(G)\equiv \chif(G,\vec 1)$, etc.
\begin{theorem}\label{thm:equalfu}
For any graph $G$ and vertex weights $\vec r\in \naturals^{\abs{V(G)}}$ we have 
$d_\pi(G^{\vec r})= d_\pi(G,\vec r)$,
$\chif(G^{\vec r})= \chif(G,\vec r)$, and
$\overline\vartheta(G^{\vec r})= \overline\vartheta(G,\vec r)$.
In addition,
$\chif(G,m \vec r)=m \chif(G,\vec r)$ and
$\overline\vartheta(G,m \vec r)=m \overline\vartheta(G,\vec r)$
hold for any $m\in \naturals$.
\end{theorem}
\noindent
The proof is provided in Appendix~\ref{app:equalfu}.
As a consequence we extend the relation \cite{lovasz1979shannon} 
$\overline\vartheta(G)\le d_\pi(G)$ (see also Appendix~\ref{app:basicgt}) to 
the case of nonunit rank,
\begin{equation}\label{eq:tdx}
 \overline\vartheta(G,\vec{r}) \leq d_\pi(G,\vec{r}).
\end{equation}
Similarly we have generalization of the condition in Eq.~\eqref{e:ramanathan}:
Whenever a graph $G$ has a rank-$\vec r$ projective representation featuring 
SIC, then it holds that
\begin{equation}\label{eq:dchi}
 d_\pi(G,\vec r)<\chif(G,\vec r).
\end{equation}

~

Following the ideas from Refs.~\onlinecite{ramanathan2014necessary, 
mancinska2016quantum}, we consider quantum models that use the maximally mixed 
state $\rho_\mathrm{mm}=\id/d$, where $d$ is the dimension of the Hilbert 
space.
For a rank-$\vec r$ projective representation, the corresponding probability 
assignment is then simply given by
\begin{equation}
 P_\mathrm{mm}(k)= r_k/d.
\end{equation}
If the representation features SIC, then $P_\mathrm{mm}\notin \stabb(G)$, 
since, by definition, $P_\mathrm{mm}\in \mathcal P_\mathrm{SIC}$ while 
$\mathcal P_\mathrm{SIC}$ and $\stabb(G)$ are disjoint sets.
This is the motivation to define the set $\rankb(G)$ of all probability 
assignments $P_\mathrm{mm}$ which arise from any projective representation of 
$G$.
That is,
\begin{equation}
 \rankb(G) =\set{\vec r/\ell |
   \vec r\in \naturals^{\abs{V(G)}} \text{, } \ell\in \naturals, \text{ such 
that } \ell\ge d_\pi(G,\vec r)}.
\end{equation}
Denoting by $\overline{\rankb}(G)$ the topological closure of $\rankb(G)$ we 
show in Appendix~\ref{app:subset} the following inclusions.
\begin{theorem}\label{thm:subset}
For any graph $G$, the set $\overline{\rankb}(G)$ is convex and
$\stabb(G) \subseteq \overline{\rankb}(G) \subseteq \thb(G)$.
\end{theorem}
\noindent
This implies that any NCHV probability assignment can be arbitrarily well 
approximated by a quantum probability assignment using the maximally mixed state.
Conversely, if $\rankb(G)\subset \stabb(G)$ for an exclusivity graph $G$, then 
any quantum probability assignment using the maximally mixed state can 
be reproduced by a NCHV model and hence no projective representation of $G$ can 
feature SIC.
This is the case for all graphs with at most 8 vertices, as we show
in Appendix~\ref{app:stabisrank} by using a linear relaxation of $\rankb(G)$.
\begin{theorem}\label{thm:stabisrank}
$\stabb(G) = \overline{\rankb}(G)$ for any graph $G$ with $8$ vertices or less.
\end{theorem}
\noindent
Since any exclusivity graph allowing SIC must have
$\overline{\rankb}(G)\supsetneq \stabb(G)$, this implies the following statement.
\begin{corollary}\label{cor:min9}
Any scenario allowing SIC requires more than 8 events.
\end{corollary}

\section{Minimal State-Independent Contextuality}\label{sec:exhaustive}

We now aim to find the smallest scenario allowing SIC, that is, the smallest 
exclusivity graph which has a projective representation featuring SIC.
Here, we say that a graph $G'$ is smaller than a graph $G$ if either $G'$ has 
less vertices than $G$ or if both have the same number of vertices and $G'$ has 
less edges than $G$.
With this notion, the smallest known graph allowing SIC is $G'_\mathrm{YO}$ 
with 13 vertices and 23 edges,\footnote{%
In fact, $G'_\mathrm{YO}$ has the same rank-one projective representation as 
$G_\mathrm{YO}$ and one can verify that the corresponding set $\mathcal{P}_{\rm SIC}$ is 
disjoined from $\stabb(G'_{\rm YO})$.}
where $G'_\mathrm{YO}$ is $G_\mathrm{YO}$ but with one edge removed as shown in 
Figure~\ref{fig:g30}~(a).
Due to Corollary~\ref{cor:min9} it remains to consider the graphs with 9 and up 
to 12 vertices as well as all graphs with 13 vertices and 23 edges or less.

Instead of testing for a projective representation featuring SIC, we use the 
weaker condition in Eq.~\eqref{eq:dchi} and we limit our considerations to 
rank-$r$ representations where all projectors have the same rank and $r=1$, 
$r=2$, or $r=3$.
We now aim to establish the following.
\begin{assertion}\label{as:main}
For $r=1,2,3$, the smallest graph $G$ with $d_\pi(G,r\vec 1)< \chif(G,r\vec 1)$ 
is $G'_\mathrm{YO}$.
\end{assertion}
\noindent
This assertion implies that $G'_\mathrm{YO}$ is the smallest graph admitting 
SIC when considering rank-$r$ projective representations for $r=1,2,3$.

Our approach to Assertion~\ref{as:main} consists of two steps.
First we identify four conditions that are easy to compute and necessary for 
$d_\pi(G,\vec r) <\chif(G, \vec r)$ to hold.
For graphs which satisfy all these conditions and for $\vec r=r\vec 1$ with 
$r=1,2,3$, we then implement a numerical optimization algorithm in order to 
compute $d_\pi(G,r\vec 1)$.
We then confirm Assertion~\ref{as:main}, aside from the uncertainty that is due 
to the numerical optimization.

\subsection{Conditions}\label{sec:conditions}

We introduce four necessary conditions that are satisfied if 
$G$ is the smallest graph with $d_\pi(G,\vec r)<\chif(G,\vec r)$ for some fixed 
$\vec r$.
First, we consider the case where $G$ is not connected.
Then there exists a partition of the vertices $V(G)$ into disjoint subsets 
$V_i\subsetneq V(G)$ such that no two vertices from different subsets are 
connected.
We write $G_i$ for the corresponding induced subgraph and similarly
 $\vec r_i$.
It is easy to see (see Appendix~\ref{app:basicgt}), that $d_\pi(G,\vec r)= 
\max_i d_\pi(G_i, \vec r_i)$ and $\chif(G,\vec r)= \max_i\chif(G_i, \vec r_i)$ 
and hence $d_\pi(G,\vec r)<\chif(G, \vec r)$ implies that already 
$d_\pi(G_i,\vec r_i)<\chif(G_i, \vec r_i)$ for some $i$.
But this is at variance with the assumption that $G$ is minimal.
Hence we have the following.
~\\
~\\
\emph{Condition 1.} $G$ is connected.
\\

Second, we consider a partition of $V(G)$ into disjoint subset $V_i\subsetneq 
V(G)$ such that any two vertices from different subsets are connected.
We have $d_\pi(G,\vec r) = \sum_i d_\pi(G_i,\vec r_i)$ and $\chif(G, \vec r) = 
\sum_i \chif(G_i, \vec r_i)$ (see Appendix~\ref{app:basicgt}) and hence
$d_\pi(G,\vec r)<\chif(G, \vec r)$ implies $d_\pi(G_i,\vec r)<\chif(G_i, \vec 
r)$ for some $i$ and thus $G$ is not minimal.
~\\
~\\
\emph{Condition 2.} $\overline G$ is connected.
\\

Third, we write $G-e$ for the subgraph with the edge $e$ removed.
Clearly, $d_\pi(G-e,\vec r)\le d_\pi(G,\vec r)$.
Thus, if $d_\pi(G,\vec r)<\chif(G,\vec r)$ and $\chif(G, \vec r)= 
\chif(G-e, \vec r)$, then we have already $d_\pi(G-e,\vec r)<\chif(G-e, 
\vec r)$ and $G$ cannot be minimal.
In order to avoid this contradiction, we need the following.
~\\
~\\
\emph{Condition 3.}
$\chif(G, \vec r) \neq \chif(G-e, \vec r)$ for all edges $e$.
\\

Note that if $\vec r= r\vec 1$, then this condition reduces to $r\chif(G)\ne 
r\chif(G-e)$ and is independent of $r$.
We can further sharpen Condition~3 by assuming merely $\ceil{\chif(G,\vec r)}= 
\ceil{\chif(G-e, \vec r)}$, where $\ceil x$ denotes the least integer not 
smaller than $x$.
Then $d_\pi(G,\vec r)<\chif(G, \vec r)$ implies $d_\pi(G-e,\vec r) < 
\ceil{\chif(G-e,\vec r)}$ and since $d_\pi(G-e, \vec r)$ is an integer, this 
also implies $d_\pi(G-e,\vec r) < \chif(G-e, \vec r)$.
~\\
~\\
\emph{Condition 4.}
$\ceil{\chif(G, \vec r)} \neq \ceil{\chif(G-e, \vec r)}$ for all edges $e$. 
\\

Finally, from Eq.~\eqref{eq:tdx} we have $\overline\vartheta(
G,\vec r) \le d_\pi(G,\vec r)$ and since $d_\pi(G,\vec r)$ is an integer, 
we also have $\ceil{\,\overline\vartheta(G,\vec r)}\le d_\pi(G,\vec r)$.
This implies our last condition.
~\\
~\\
\emph{Condition 5.}
$\ceil{\,\overline\vartheta(G, \vec r)} < \chif(G, \vec r)$.
\\

We apply these five conditions 
to all graphs with $n=9,10,11,12$ vertices and all graphs with $n=13$ vertices 
and 23 or less edges.
The resulting numbers of graphs are listed in Table~\ref{tab:graph_conditions}.
First, all nonisomorphic graphs are generated using the software package 
``nauty'' \cite{mckay2014practical}, where then all graphs violating 
Condition~1 or Condition~2 are discarded.
Subsequently, Condition~3 is implemented and for the remaining graphs, 
$\overline\vartheta(G)$, $\chif(G)$, and $\min_e \chif(G-e)$ are computed, 
which then allows us to evaluate Condition~4 and Condition~5 for $\vec r=r\vec 
1$ with $r=1,2,3$.

For the computation of $\chif$, we use a floating point solver for the 
corresponding linear program.
On the basis of the solution of the program, an exact fractional solution is 
guessed and then verified using the strong duality of linear optimization.
The Lovász number $\vartheta$ is computed by means of a floating point solver 
for the corresponding semidefinite program.
The dual and primal solutions are verified and the gap between both is used to 
obtain a strict upper bound on the numerical error.
This error is in practice of the order of $10^{-10}$ or better for the vast 
majority of the graphs.

\begin{table}
\begin{tabular}{c c|r r r r r r}
\hline\hline
~ rank ~ & Condition &
\multicolumn{1}{c}{$n=9$} & \multicolumn{1}{c}{$n=10$} &
\multicolumn{1}{c}{$n=11$} & \multicolumn{1}{c}{$n=12$} &
\multicolumn{1}{c}{$n=13^*$} \\
\hline
any & none &
$274\,668$ & $12\,005\,168$ & $1\,018\,997\,864$ & $165\,091\,172\,592$ & 
$10\,951\,875\,086$\\
any & 1 \& 2 &
$247\,492$ & $11\,427\,974$ & $994\,403\,266$ & $163\,028\,488\,360$ & 
$9\,185\,079\,351$\\
$\vec r= r\vec1$ & 1--3 & $52$ & $608$ & $13\,716$ & $609\,373$ & $16\,893$\\
$r=1$ & 1--4 & $37$ & $283$ & $5\,122$ & $163\,127$ & $15\,596$\\
$r=1$ & 1--5 & $1$ & $11$ & $446$ & $31\,049$ & $77$\\
$r=2$ & 1--4 & $44$ & $398$ & $7\,159$ & $238\,478$ & $15\,691$\\
$r=2$ & 1--5 & $8$ & $126$ & $2\,483$ & $106\,400$ & $172$\\
$r=3$ & 1--4 & $45$ & $430$ & $8\,240$ & $265\,346$ & $15\,865$\\
$r=3$ & 1--5 & $13$ & $158$ & $3\,574$ & $133\,268$ & $346$ \\
\hline\hline
\end{tabular}
\caption{\label{tab:graph_conditions}%
Numbers of graphs satisfying Condition~1--5.
Condition~1--5 are applied to all graphs with $n$ vertices.
For $n=13^*$ vertices, only the graphs with up to 23 edges are considered.
Condition~3 can be applied for the case $\vec r=r\vec 1$ for any $r$ but 
Condition~4 and Condition~5 are evaluated only for the cases $r=1,2,3$.
}
\end{table}

\subsection{Numerical estimate of the dimension}\label{sec:seesaw}
If an exclusivity graph $G$ has a rank-$\vec r$ projective representation with 
SIC, then, according to Theorem~\ref{thm:equalfu} and the subsequent 
discussion, there must be a rank-one projective representation of $G^\vec r$ in 
dimension $d=\lceil\chif(G,\vec r)\rceil-1$.
At this point, we do not further exploit the structure of the problem.
We rather consider methods which allow us to verify or falsify the existence 
of a rank-one projective representation in dimension $d$ of an arbitrary graph 
$G$ with $n$ vertices.

If such a projective representation exists, then one can assign normalized 
vectors $\vec y_k\in \complexes^d$ to each vertex $k\in V(G)$ such that $\vec 
y_\ell^\dag \vec y_k=0$ for all edges $[\ell, k]\in E(G)$.
Collecting these vectors in the columns of a matrix $Y$, we obtain the 
feasibility problem
\begin{equation}
\label{eq:feasible}
\begin{array}{lcl}
\text{find} &\quad& X=Y^\dag Y \text{ with } Y\in \complexes^{d\times n},\\
\text{subject to} && X_{k,k} = 1 \text{ for all } j\in V(G),\\
&& X_{\ell,k} = 0 \text{ for all } [\ell,k]\in E(G).\\
\end{array}
\end{equation}
This problem is equivalent to the optimization problem
\begin{equation}
\label{eq:mini}
\begin{array}{lcl}
\text{minimize} &\quad&
\sum_{k\in V(G)} (X_{k,k}-1)^2+ \sum_{[\ell,k]\in E(G)} X_{\ell,k}^2,\\
\text{with} && X=Y^\dag Y \text{ and } Y\in \complexes^{d\times n},
\end{array}
\end{equation}
where the problem in Eq.~\eqref{eq:feasible} is feasible if and only if the 
problem in Eq.~\eqref{eq:mini} yields zero.
The optimization can be executed using a standard algorithm like the  
conjugate-gradient method \cite{press2007numerical}.
However, the obtained value can be from a local minimum and depend on the 
 initial value used in the optimization.
Hence obtaining a value greater than zero does not conclusively exclude the 
existence of a projective representation, but this problem can be mitigated by 
performing the minimization for many different initial values.

Instead of employing one of the standard optimization algorithms, we use a 
faster method that allows us to repeat the minimization with many different 
initial values.
For this we denote by $\mathcal L$ the set of all $(n\times n)$-matrices $X$ 
which satisfy the constraints of the problem in Eq.~\eqref{eq:feasible} and we 
write $\mathcal R$ for the set of all matrices $X$ for which $X=Y^\dag Y$ for 
some $(d\times n)$-matrix $Y$.
In an alternating optimization, we generate a sequence $(X^{(j)})_j$ from an 
initial value $X^{(0)}$ such that
\begin{equation}\begin{split}
 X^{(2i+1)}&= \argmin_R \set{\norm{R-X^{(2i)}} | R\in \mathcal R},\\
 X^{(2i)}&= \argmin_L \set{ \norm{L-X^{(2i-1)}} | L\in \mathcal L}.
\end{split}\end{equation}
By construction, $\delta_j=\norm{X^{(j)}-X^{(j-1)}}$ is a nonincreasing 
sequence and hence $\delta_\infty=\lim_{j\to \infty} \delta_j$ exists.
Consequently, for the existence of a projective representation it is sufficient 
if $\delta_\infty=0$ because then $X^{(\infty)}=\lim_{j\to\infty} X^{(j)}$ 
exists with $X^{(\infty)} \in \mathcal R\cap\mathcal L$.
In Appendix~\ref{app:seesaw} we show that this alternating optimization can be 
implemented efficiently for the Frobenius norm $\norm{M}_F=\sum_{i,j} 
\abs{M_{i,j}}^2$.

We run the optimization with $100$ randomly chosen initial values $X^{(0)}$ for each of 
the remaining graphs with corresponding rank $r$.
We stop the optimization if $\delta_{k-2}/\delta_k < 1+10^{-5}$.
For all graphs and all repetitions the optimization converges with a final 
value of $\delta_k$ in the order of $1$.
In comparison, we test the algorithm for many graphs with known $d_\pi$ where 
the graphs have up to 40 vertices.
In all these cases, the algorithm converges to $\delta_k$ in the order of 
$10^{-9}$, which gives us confidence that the alternating optimization is 
reliable.
In summary this constitutes strong numerical evidence that none of the 
remaining graphs with corresponding rank has a projective representation with 
SIC.

\section{Conclusion and discussion}\label{sec:conclusion}

The search for a primitive entity of contextuality has not yet reached a 
conclusion despite of decades of research on this topic.
Of course, one can argue that the pentagon scenario by Klyachko \emph{et al.} 
\cite{klyachko2008simple} does provide a provably minimal scenario.
But the drawback of the pentagon scenario is that it is state-dependent.
That is, contextuality is here a feature of both, the state and the 
measurements.
In contrast, in the state-independent approach, contextuality is a feature 
exclusively of the measurements and we argue that a primitive entity of 
contextuality should embrace state-independence.
Among the known SIC scenarios, the one by Yu and Oh 
\cite{yu2012stateindependent} is minimal and this has also been proved 
rigorously for the case where all measurement outcomes are represented by rank-one projectors.

As we pointed out here, there is no guarantee that the actual minimal scenario 
will also be of rank one:
We showed that a scenario by Toh\cite{toh2013stateindependent}---albeit far 
from minimal---requires projectors of rank two.
This motivated our search for the minimal SIC scenario for the case of nonunit 
rank.
Due to Theorem~\ref{thm:stabisrank}, we can exclude the case where the 
exclusivity graph has 8 or less vertices.
For the remaining cases of 9 to 12 vertices, we also obtain a negative result, 
however, under the restriction that the projective representation is uniformly 
of rank two or uniformly of rank three.
A key to this result is a fast and empirically reliable numerical method to 
find or exclude projective representations of a graph, which might be also a 
useful method for related problems in graph theory.

Curiously, there is no simple argument that shows that the scenario by Yu and 
Oh is minimal, even when assuming unit rank.
This in contrast to the case of state-dependent contextuality, where the reason 
that the pentagon scenario is the simplest scenario beautifully has the origin 
in graph theory\cite{cabello2014graphtheoretic}.
For the future it will be interesting to develop additional methods for SIC, in 
particular for the case of heterogeneous rank.
It will be particularly interesting whether this problem can be solved using 
more methods from graph theory, whether it can be solved using new numerical 
methods, or whether the problem turns out to be genuinely hard to decide.

\begin{acknowledgements}
We thank A.\ Cabello, N.\ Tsimakuridze, Y.Y.\ Wang for discussions, A.\ Ganesan 
for pointing out Ref.~\onlinecite{schrijver2004fractional} to us, and the 
University of Siegen for enabling our computations through the HoRUS cluster.
This  work  was  supported  by the Deutsche Forschungsgemeinschaft  (DFG,  
German  Research Foundation - 447948357), the ERC (ConsolidatorGrant  
683107/TempoQ), and the Alexander von Humboldt Foundation.
\end{acknowledgements}

\appendix

\section{Elements from graph theory}
\label{app:basicgt}
A graph $G$ is a collection of vertices $V(G)$ connected by edges $E(G)$. Each 
edge $[i,j]\in $ is an unordered pair of the vertices $i\ne j\in V(G)$. 
Conversely, for a given vertex set $V$ and edge set $E$ the pair $(V,E)$ forms 
the graph denoted by $G(V,E)$. For a given subset $W$ of $V$ and subset $F$ of 
$E$, the graph $G(W,F)$ is a subgraph of $G(V,E)$. In the case where
\begin{equation}
 F = E \cap \{[i,j]\}_{i,j \in W},
\end{equation}
$G(W,F)$ is a subgraph of $G(V,E)$ induced by the subset $W$. In the case where
\begin{equation}
 F = \{[i_t,i_{t+1}]\}_t,
\end{equation}
$G(W,F)$ is  a path in $G(V,E)$. A graph $G(V,E)$ is connected if any two 
vertices can be connected by a path.
A subset of vertices $C$ is a clique, if in the induced subgraph all vertices 
are mutually connected by an edge.
A clique $C$ is maximal, if any strict superset of $C$ is not clique.
The complement graph $\bar G$ of $G$ has an edge $[i,j]$ if and only if $i\ne 
j$ and $[i,j]$ is not an edge in $G$.
A clique in $\bar G$ is an independent set of $G$.
Independent sets are also called stable sets.
If any strict superset of $W$ is not an independent set, then $W$ is a 
maximally independent set.

Now, the index vector of a given subset of vertices $W$ is defined as
\begin{equation}
  \Delta_W = [\delta_W(k)]_{k\in V},
\end{equation}
where $\delta_W(k)=1$ if $k\in W$ and $\delta_W(k)=0$ otherwise. Let 
$\mathcal{I}$ denote the set of all independent sets of  graph $G$, then the 
stable set polytope $\stabb(G)$ is the convex hull of the set $\set{\Delta_W | 
W \in \mathcal{I}}$.

A collection of real vectors $(\vec{v}_i)_{i\in V}$ is an orthogonal 
representation (OR) of $G$, provided that
$[i,j] \not\in E$  implies $\vec{v}_i\cdot \vec{v}_j = 0$.
The Lovász theta body of a given graph $G$ can be defined 
as~\cite{grotschel1986relaxations}
\begin{equation}
  \thb(G) = \set{[(\vec{s}\cdot\vec{v}_i)^2]_{i\in V} | (\vec{v}_i)_{i\in V}\ 
\text{is an OR of }\overline{G}},
\end{equation}
where $\vec{s} = (1,0,\ldots,0)$. We also use the following, equivalent 
definition of $\thb(G)$. A collection of projectors $(\Pi_k)_{k\in V}$ (over a 
complex Hilbert space) is a projective representation (PR) of $G$ if 
$\Pi_i\Pi_j=0$ whenever $[i,j] \in E(G)$.
Then, one can also write \cite{cabello2014graphtheoretic}
\begin{equation}\label{eq:defthb}
  \thb(G) = \set{[\tr(\rho \Pi_i)]_{i\in V}| (\Pi_i)_{i\in V}\ \text{is a PR of 
}G, \tr(\rho)=1, \rho\ge 0}.
\end{equation}
Note that in the definition, the projectors might be of any rank.

For a vector $\vec{r}$ of nonnegative real numbers,
\begin{equation}
\alpha(G,\vec{r}) = \max_{\vec{x}}
 \set{\vec{r}\cdot\vec{x}|\vec{x} \in \stabb(G)}
\end{equation}
is the weighted independence number \cite{grotschel1986relaxations} and the 
weighted Lovász number is given \cite{knuth1994sandwich} by
\begin{equation}\label{eq:thTH}
\vartheta(G,\vec{r}) = \max_{\vec{x}}
 \set{\vec{r}\cdot \vec{x}|\vec x\in \thb(G)}.
\end{equation}
For convenience, we write $\overline{\vartheta}(G,\vec{r}) = 
\vartheta(\overline{G},\vec{r})$.

The weighted chromatic number $\chi(G,\vec{r})$ can be defined 
as\cite{schrijver2004fractional}
\begin{equation}\label{eq:chromatic}
\begin{aligned}
\min_{(c_I)_{I\in \mathcal I}}&\quad&& \sum_{I\in \mathcal{I}} c_I,\\
\text{such that}&&& \sum_{I\ni i} c_I \ge r_i, \text{ for all } i\in V,
\end{aligned}
\end{equation}
where $c_I$ are nonnegative integers.
Equivalently, if $C=\chi(G,\vec r)$, then there exists an $\vec r$-coloring of 
$G$ with $C$ colors, that is, $C$ is the minimal number of colors such that 
$r_k$ colors are assigned to each vertex $k$ and two vertices $i$ and $j$ do 
not share a common color if they are connected.

The weighted fractional chromatic number $\chif(G,\vec r)$ is a relaxation of 
the integer program in Eq.~\eqref{eq:chromatic} to a linear program 
\cite{schrijver2004fractional}
\begin{equation}\label{eq:defchif}
\begin{aligned}
\min_{(x_I)_{I\in \mathcal I}}&\quad&& \sum_{I\in \mathcal{I}} x_I,\\
\text{such that}&&& \sum_{I\ni i} x_I \ge r_i, \text{ for all } i\in V,
\end{aligned}
\end{equation}
where $x_I$ are now nonnegative real numbers.
Being a linear program with rational coefficients, all $x_I$ can be chosen to 
be rational numbers and hence one can find a $b\in \mathbb N$ such that all 
$bx_I$ are integer.
This yields the relation
\begin{equation}
  \chi_f(G,\vec{r}) = \min_{b\in \mathbb{N}} \frac{\chi(G,b\vec{r})}{b}.
\end{equation}

Finally, we use $d_\pi(G,\vec r)$ as defined in the main text, that is, 
$d_\pi(G,\vec r)$ is the minimal dimension admitting a rank-$\vec r$ PR.
We also omit the weights $\vec r$ for the functions $d_\pi$, $\chif$, and 
$\overline\vartheta$, if $\vec r=\vec 1$.
We now show the known relation \cite{lovasz1979shannon} 
$\overline\vartheta(G)\le d_\pi(G)$, which is extended to the case of $\vec 
r=\vec 1$ in Eq.~\eqref{eq:tdx} in the main text.
\begin{lemma}
 $\overline\vartheta(G)\le d_\pi(G)$
\end{lemma}
\begin{proof}
For a given $d$-dimensional rank-$1$ PR $(\Pi_k)_k$ of $G$, a $d^2$-dimensional rank-$1$ PR $(P_k)_k$ of $G$ can be constructed as
\begin{equation}
P_k = \Pi_k^* \otimes \Pi_k,
\end{equation}
where complex conjugation is with respect to some arbitrary, but fixed 
orthonormal basis $\ket 1,\ket 2,\dotsc,\ket d$.
Using $\Psi = \sum_{j,l} \ketbra{jj}{ll}$, we have
$\tr(\Psi P_k) = 1$ and $\tr(\Psi)=d$.

We consider now an arbitrary rank-$1$ PR $(Q_k)_k$ of $\overline{G}$ together 
with an arbitrary density operator $\rho$ acting on the same Hilbert space as 
the PR. Then $(P_i\otimes Q_j)_{i,j}$ is a PR of $G\otimes \overline{G}$ and 
$(i,i)$ is connected with $(j,j)$ either within $G$ or within $\overline G$, 
for any two vertices $i\ne j$.
Here $G\otimes K$ denotes the graph with vertices $V(G)\times V(K)$ and where  
$[(v,w),(v',w')]$ is an edge, if $[v,v']$ or $[w,w']$ is an edge.

Therefore, $\sum_k P_k\otimes Q_k\le \openone$ and consequently,
\begin{equation}
d = \tr(\Psi\otimes \rho) \ge \sum_{k} \tr([\Psi\otimes \rho][P_k\otimes 
Q_k]) = \sum_{k} \tr(\rho Q_k).
\end{equation}
By virtue of Eq.~\eqref{eq:defthb} we obtain $\sum x_i\le d$ for all $\vec x\in 
\thb(\overline G)$, which then yields the desired inequality due to 
Eq.~\eqref{eq:thTH}.
\end{proof}

The disjoint union $G=G_1\cup G_2$ of two graphs consists of the disjoint union 
of the vertices, $V(G)=V(G_1)\uplus V(G_2)$, and $[i,j]$ is an edge in $G$ if 
it is an edge in either $G_1$ or $G_2$. For Condition~1 in 
Section~\ref{sec:conditions} we use the following observation.
\begin{lemma}
If $G=\bigcup_i G_i$, then $d_\pi(G,\vec r)= \max_i d_\pi(G_i, \vec r_i)$ and 
$\chi_f(G,\vec r)= \max_i \chi_f(G_i, \vec r_i)$, where $\vec{r}_i$ is the part 
of $\vec{r}$ for $G_i$.
\end{lemma}
\begin{proof}
By definition, $d_\pi(G,\vec{r}) \ge \max_i d_\pi(G_i,\vec{r}_i)$. Conversely, 
if $d =\max_i d_\pi(G_i,\vec{r}_i)$ then we can find a $d$-dimensional 
rank-$\vec{r}_i$ PR for each $G_i$. Since the subgraphs are mutually disjoined, 
these PRs jointly form already a $d$-dimensional rank-$\vec{r}$ PR of $G$. Thus 
$d \ge d_\pi(G,\vec{r})$.

For the fractional chromatic number, one first observes that $G^{\vec 
r}=\bigcup_i G_i^{\vec r_i}$. Hence the assertion reduces to $\chif(\bigcup_i 
G_i^{\vec r_i})=\max_i\chif(G_i^\vec{r_i})$, which is a well-known relation for 
disjoint unions of graphs \cite{scheinerman1997fractional}.
\end{proof}

The join $G=G_1+ G_2$ of two graphs is similar to the disjoint union, however 
with an additional edge between any two vertices $[i,j]$ if $i\in V(G_1)$ and 
$j\in V(G_2)$. For Condition~2 in Section~\ref{sec:conditions} we use then the 
following observation.
\begin{lemma}
If $G=\sum_i G_i$, then
$d_\pi(G, \vec r) = \sum_i d_\pi(G_i, \vec r_i)$ and $\chi_f(G, \vec r) = 
\sum_i \chi_f(G_i, \vec r_i)$.
\end{lemma}
\begin{proof}
For given $d_i$-dimensional rank-$\vec{r}_i$ PRs $(\Pi_j)_{j\in V(G_i)}$ of 
$G_i$, we define
\begin{equation}
  P_{j,i} = \left(\bigoplus_{k<i} \oo_{k}\right) \oplus \Pi_j \oplus 
\left(\bigoplus_{k>i} \oo_{d_k}\right),
\end{equation}
where $j \in G_i$ and $\oo_k$ is the zero-operator acting on the space of the 
PR of $G_k$. This construction achieves that $((P_{j,i})_{j\in V(G_i)})_i$ is a 
$(\sum_i d_i)$-dimensional rank-$\vec{r}$ PR of $G$ and therefore 
$d_\pi(G,\vec{r}) \le \sum_i d_\pi(G_i, \vec{r}_i)$ holds. Conversely, from a 
given $d$-dimensional rank-$\vec{r}$ PR of $G$, we can deduce a 
$d_i$-dimensional rank-$\vec{r}_i$ PR of each $G_i$, where $d_i$ is the 
dimension of the subspace $S_i$ where  $(\Pi_j)_{j\in G_i}$ acts nontrivially. 
Since each of subspace $S_i$ is orthogonal to the other subspaces $S_j$, we 
obtain $d \ge \sum_i d_i \ge \sum_i d_\pi(G_i, \vec{r}_i)$.

For the fractional chromatic number, we note that $G^{\vec r}=\sum_i G_i^{\vec 
r_i}$ and since $\chif$ is additive under the join of graphs 
\cite{scheinerman1997fractional}, the assertion follows.
\end{proof}

\section{$G_\mathrm{Toh}$ has no rank-one projective representation}
\label{app:gtoh}
It can be verified numerically that there is no $4$-dimensional rank-$1$ PR of 
$G_{\rm Toh}$ with our numerical methods in  Appendix~\ref{app:seesaw}. Here, 
we give an analytical proof with the help of the computer algebra system 
Mathematica.

Since each (row) vector $\vec{v}$ corresponds to a rank-$1$ projector 
$P(\vec{v}) = \vec{v}^\dagger\vec{v}/|\vec{v}|^2$, we can use vectors instead 
of projectors in the case of rank-$1$ PR. Also, two non-zero vectors 
$\vec{v}_1$ and $\vec{v}_2$ are called equal if $P(\vec{v}_1) = P(\vec{v}_2)$.  
For three independent vectors $\vec{v}_1, \vec{v}_2, \vec{v}_3$ in the 
$4$-dimensional Hilbert space, from Cramer's rule we know that their common 
orthogonal vector is proportional to $\Lambda(\vec{v}_1, \vec{v}_2, \vec{v}_3) 
= (\lambda_1, \lambda_2, \lambda_3, \lambda_4)^*$, with $\vec{v}_i = \{v_{i,1}, 
v_{i,2}, v_{i,3}, v_{i,4}\}$,
\begin{equation}
  \lambda_i = (-1)^i \left|\begin{array}{ccc}
	v_{1,i+1} & v_{1,i+2} & v_{1,i+3}\\
        v_{2,i+1} & v_{2,i+2} & v_{2,i+3}\\
        v_{3,i+1} & v_{3,i+2} & v_{3,i+3}\\
   \end{array}\right|,\
\end{equation}
where the sum $i+j$ is modulo $4$.
The proof that there is no $4$-dimensional rank-$1$ SIC set for $G_{30}$ can be 
divided in two cases.

\textbf{Case~1:} Let $\{\vec{v}_i\}_{i\in V(G_{30})}$ be a $4$-dimensional 
rank-$1$ PR. We first consider the case where 
\begin{equation}
  \vec{v}_i\neq\vec{v}_j \text{~~~for~} (i,j) \in \{(5,21), (4,24), (14,23), 
  (3,10), (3,22), (4,22)\}.
  \label{eq:case1}
\end{equation}
We can have the following process of parametrization in the basis of 
$\{\vec{v}_{28}, \vec{v}_{14}, \vec{v}_1, \vec{v}_{22}\}$:
\begin{eqnarray}
  &\vec{v}_{28} = (1,0,0,0); \vec{v}_{14} = (0,1,0,0); \vec{v}_1 = (0,0,1,0); \vec{v}_{22} = (0,0,0,1);\\ 
  &\vec{v}_{2}=(\cos x_1,\sin x_1,0,0);\vec{v}_{13}=(0,0,\cos x_2,e^{i \theta_1} \sin x_2);\\
  &\vec{v}_{17}=(-\sin x_1 \cos x_3,\cos x_1 \cos x_3,e^{i \theta_2} \sin x_3,0);\\
  &\vec{v}_{29}=(-\sin x_1 \sin x_3,\cos x_1 \sin x_3,-e^{i \theta_2} \cos x_3,0);\\
  &\vec{v}_{20}=(\cos x_4\cos x_1,\cos x_4\sin x_1,0,\sin x_4);\\
  &\vec{v}_{3}=(-\sin x_1 \cos x_5,\cos x_1 \cos x_5,0,e^{i \theta_3} \sin x_5);\\
  &\vec{v}_{4}=(-\sin x_1 \sin x_5,\cos x_1 \sin x_5,0,-e^{i \theta_3} \cos x_5).
\end{eqnarray}
We claim that $\vec{v}_3$ is not on the plane spanned by $\vec{v}_{20}, 
\vec{v}_{29}$, otherwise $\vec{v}_5 \perp \vec{v}_3$. Thus, $\vec{v}_5 
= \vec{v}_{21}$ since they are orthogonal to $\vec{v}_3, \vec{v}_6, 
\vec{v}_{24}$ in the $4$-dimensional space. This is conflicted with the 
assumption in Eq.~\eqref{eq:case1}.  Hence, we get $\vec{v}_{24} 
= \Lambda(\vec{v}_3, \vec{v}_{20}, \vec{v}_{29})$, which further leads to 
$\vec{v}_5 = \Lambda(\vec{v}_{20},\vec{v}_{24},\vec{v}_{29})$.  Note that 
$\vec{v}_4 \perp \vec{v}_3, \vec{v}_{4}\perp \vec{v}_{21}$, hence $\vec{v}_4$ 
is on the plane spanned by $\vec{v}_6, \vec{v}_{24}$. Since $\vec{v}_4 \neq 
\vec{v}_{24}$, we get $\vec{v}_6 = (\vec{v}_{24}\vec{v}_{24}^\dagger)\vec{v}_4 
- (\vec{v}_4\vec{v}_{24}^\dagger)\vec{v}_{24}$ and hence $\vec{v}_{21} 
= \Lambda(\vec{v}_3,\vec{v}_6,\vec{v}_{24})$. Since $\vec{v}_{14} \neq 
\vec{v}_{23}$, we have that $\vec{v}_7 
= \Lambda(\vec{v}_5,\vec{v}_6,\vec{v}_{14})$, $\vec{v}_8 = \Lambda(\vec{v}_5, 
\vec{v}_6, \vec{v}_7)$, and $\vec{v}_{23} = \Lambda(\vec{v}_8, \vec{v}_{13}, 
\vec{v}_{28})$. Since $\vec{v}_3 \neq \vec{v}_{10}$, we have that $\vec{v}_9 
= \Lambda(\vec{v}_3,\vec{v}_7,\vec{v}_{23})$, 
$\vec{v}_{27}=\Lambda(\vec{v}_7,\vec{v}_9,\vec{v}_{23})$, and 
$\vec{v}_{10}=\Lambda(\vec{v}_4,\vec{v}_{21},\vec{v}_{27})$.

For the following proof, we make use of the computer algebra system 
Mathematica.  Since $\vec{v}_4 \perp \vec{v}_5$, direct computation shows that 
$\sin (2 x_5) \sin (x_3-x_4) \sin (x_3+x_4) = 0$. As $\sin (2 x_5) = 0$ will 
result in either $\vec{v}_3 = \vec{v}_{22}$ or $\vec{v}_4 = \vec{v}_{22}$, 
which conflicts with the assumption in Eq.~\eqref{eq:case1}, we have that 
$x_3=\pm x_4 \mod \pi$. Because of the freedom of choosing $\theta_2$, we can, 
without loss of generality, assume that $x_3 = x_4$.  Then $|\vec{v}_8|^2 > 0$ 
implies that $\sin x_4\cos x_4 \neq 0$.  Further, $\vec{v}_8 \perp 
\vec{v}_{28}$ implies that $\cos^2 x_1 = e^{2i\theta_3}\sin^2 x_1$, i.e., 
$\theta_3=0 \mod \pi$ and $x_1 = \pm \pi/4 \mod \pi$. Without loss of 
generality, we can assume that $x_1 = \pi/4$.  Then $|\vec{v}_8|^2 > 0$ also 
implies that $\sin(x_4+x_5) \neq 0$. Since $\vec{v}_7 \perp \vec{v}_{23}, 
\vec{v}_8 \perp \vec{v}_{13}$, we can find that $\cos x_2+e^{i 
(\theta_1+\theta_2)} \sin x_2 = 0$. Without loss of generality, we can assume 
$x_2 = -\pi/4, \theta_1=-\theta_2$. All the above arguments result in that
\begin{equation}
  \vec{v}_8\vec{v}_{10}^\dagger = -e^{i\theta_2} \sin^{10}x_4 \cos^{29}x_4 
  \sin^5(x_4+x_5)/\sqrt{2} \neq 0,
\end{equation}
which conflicts with the exclusivity relations. Thus, $\vec{v}_i = \vec{v}_j$ 
should hold for at least one pair of $(i,j) \in 
\{(5,21),(4,24),(14,23),(3,10),(3,22),(4,22)\}$.

\textbf{Case~2:} Let $\{\vec{a},\vec{b},\vec{c},\vec{d}\}$ be an orthogonal 
basis and $\vec{x},\vec{y}$ are another two vectors in the $4$-dimensional 
space, then
\begin{enumerate}
  \item $\vec{x}\perp \vec{a}, \vec{x}\perp \vec{b}, \vec{y}\perp \vec{c}, 
  \vec{y}\perp \vec{d}$ implies that $\vec{x}\perp \vec{y}$;
  \item $\vec{x}\perp \vec{a}, \vec{x}\perp \vec{b}, \vec{x}\perp \vec{c}$ 
  implies that $\vec{x} = \vec{d}$;
  \item $\vec{x}\perp \vec{a}, \vec{x}\perp \vec{b}, \vec{y}\perp \vec{x}, 
  \vec{y}\perp \vec{d}$ implies that either $\vec{x}\perp \vec{c}$ or $\vec{y} \perp 
  \vec{c}$.
\end{enumerate}
In the language of graph theory, if a given graph $G$ has a rank-$1$ PR in 
dimension $4$, then the graph obtained from  the following rules should also 
have rank-$1$ PR in dimension $4$: let $\{a,b,c,d\}$ be a clique and $x,y$ are 
two other vertices in $G$,
\begin{enumerate}
  \item if $(x,a), (x,b), (y,c), (y,d) \in E(G)$, then add $(x,y)$ to $E(G)$;
  \item if $(x,a), (x,b), (x,c) \in E(G)$, then combine $x,d$ into one vertex 
  whose neighbors is the union of the ones of $x$ and the ones of $d$;
  \item if $(x,a), (x,b), (x,y), (y,d) \in E(G)$, then add either $(y,c)$ or $(x,c)$ to 
  $E(G)$.
\end{enumerate}

When we apply these rules repeatedly to $G_{30}$ after combining any  pair in 
$\{ (5,21), (4,24),$ $(14,23), (3,10), (3,22), (4,22) \}$, we either end up 
with a graph which contains a clique with size larger than $4$ or a self-loop.  
This can be done automatically again with Mathematica. It is obvious that 
a clique of size larger than $4$ has no PR in $4$-dimensional space and 
a self-loop has no rank-1 PR.
\begin{table}[!h]
		\centering
\begin{tabular}{CCCCCCCCCCCCCCCCCCCCCCCCCCCCCCCCCCCCCCCC}
		\hline\hline
1 & 0 & 0 & 0 & 0 & 0 & 0 & 0 & 1 & 1 & 1 & 1 & 0 & 0 & 0 & 0 & 1 & 1 & 1 & 1 & 0 & 0 & 0 & 0 & 1 & 1 & 1 & 1 & 0 & 0 & 0 & 0 & 1 & 1 & 1 & 1 & 0 & 0 & 0 & 0 \\
0 & 1 & 0 & 0 & 0 & 0 & 0 & 0 & 1 & 1 & \bar{1} & \bar{1} & 0 & 0 & 0 & 0 & 1 & 1 & \bar{1} & \bar{1} & 0 & 0 & 0 & 0 & 0 & 0 & 0 & 0 & 1 & 1 & 1 & 1 & 0 & 0 & 0 & 0 & 1 & 1 & \bar{1} & \bar{1} \\
0 & 0 & 1 & 0 & 0 & 0 & 0 & 0 & 1 & \bar{1} & 1 & \bar{1} & 0 & 0 & 0 & 0 & 0 & 0 & 0 & 0 & 1 & 1 & 1 & 1 & 1 & 1 & \bar{1} & \bar{1} & 0 & 0 & 0 & 0 & 0 & 0 & 0 & 0 & 1 & \bar{1} & 1 & \bar{1} \\
0 & 0 & 0 & 1 & 0 & 0 & 0 & 0 & 1 & \bar{1} & \bar{1} & 1 & 0 & 0 & 0 & 0 & 0 & 0 & 0 & 0 & 1 & 1 & \bar{1} & \bar{1} & 0 & 0 & 0 & 0 & 1 & 1 & \bar{1} & \bar{1} & 1 & \bar{1} & 1 & \bar{1} & 0 & 0 & 0 & 0 \\
0 & 0 & 0 & 0 & 1 & 0 & 0 & 0 & 0 & 0 & 0 & 0 & 1 & 1 & 1 & 1 & 1 & \bar{1} & 1 & \bar{1} & 0 & 0 & 0 & 0 & 1 & \bar{1} & 1 & \bar{1} & 0 & 0 & 0 & 0 & 0 & 0 & 0 & 0 & \bar{1} & 1 & 1 & \bar{1} \\
0 & 0 & 0 & 0 & 0 & 1 & 0 & 0 & 0 & 0 & 0 & 0 & 1 & 1 & \bar{1} & \bar{1} & 1 & \bar{1} & \bar{1} & 1 & 0 & 0 & 0 & 0 & 0 & 0 & 0 & 0 & 1 & \bar{1} & 1 & \bar{1} & 1 & 1 & \bar{1} & \bar{1} & 0 & 0 & 0 & 0 \\
0 & 0 & 0 & 0 & 0 & 0 & 1 & 0 & 0 & 0 & 0 & 0 & 1 & \bar{1} & 1 & \bar{1} & 0 & 0 & 0 & 0 & 1 & \bar{1} & 1 & \bar{1} & 1 & \bar{1} & \bar{1} & 1 & 0 & 0 & 0 & 0 & \bar{1} & 1 & 1 & \bar{1} & 0 & 0 & 0 & 0 \\
0 & 0 & 0 & 0 & 0 & 0 & 0 & 1 & 0 & 0 & 0 & 0 & 1 & \bar{1} & \bar{1} & 1 & 0 & 0 & 0 & 0 & 1 & \bar{1} & \bar{1} & 1 & 0 & 0 & 0 & 0 & 1 & \bar{1} & \bar{1} & 1 & 0 & 0 & 0 & 0 & 1 & 1 & 1 & 1 \\
 \hline\hline
\end{tabular}
\caption{\label{tab:kp40rays} Kernaghan and Peres' $40$ rays, where the ray $\vec{r}_i$ is represented by the vector in the  $i$-th column and $\bar{1}$ 
  stands for  $-1$.}
\end{table}

The rank-$2$ PR of $G_{30}$ is made up of Kernaghan 
and Peres' $40$ rays as shown in Table \ref{tab:kp40rays}.
Denote $P_{\{a,b\}} := \vec{r}_a^\dagger\vec{r}_a/|\vec{r}_a|^2 + 
\vec{r}_b^\dagger\vec{r}_b/|\vec{r}_b|^2$ where $\vec{r}_a \vec{r}_b = 0$, then 
the vertices from $v_1$ to $v_{30}$ are represented by the following rank-$2$ 
projectors:
\begin{eqnarray}
  &P_{\{1,7\}},P_{\{2,8\}},P_{\{3,4\}},P_{\{5,6\}},P_{\{9,12\}}, P_{\{13,16\}},P_{\{14,10\}},P_{\{15,11\}},\nonumber\\
  &P_{\{19,20\}},P_{\{21,22\}},P_{\{23,17\}},P_{\{24,18\}},P_{\{28,27\}},P_{\{30,29\}},P_{\{31,25\}},\nonumber\\
  &P_{\{32,26\}},P_{\{33,35\}},P_{\{34,40\}},P_{\{36,37\}},P_{\{38,39\}},P_{\{1,2\}},P_{\{3,5\}},P_{\{9,13\}},\nonumber\\
  &P_{\{14,15\}},P_{\{19,21\}},P_{\{28,30\}},P_{\{23,24\}},P_{\{31,32\}},P_{\{34,36\}},P_{\{33,38\}}.
\end{eqnarray}

\section{Proof of Theorem~\ref{thm:equalfu}}
\label{app:equalfu}
The theorem consists of the following statements for any graph $G$, vertex 
weights $\vec r\in \naturals^{\abs{V(G)}}$, and $m\in \naturals$.
(i) $d_\pi(G^{\vec r},1)= d_\pi(G,\vec r)$,
(ii) $\chif(G^{\vec r},1)= \chif(G,\vec r)$,
(iii) $\overline\vartheta(G^{\vec r},1)= \overline\vartheta(G,\vec r)$,
(iv) $\chif(G,m \vec r)=m \chif(G,\vec r)$, and
(v) $\overline\vartheta(G,m \vec r)=m \overline\vartheta(G,\vec r)$.

\vspace{1ex}\noindent (i)
In the main text, above Theorem~\ref{thm:equalfu}, it was already shown, that 
any rank-one PR of $G^{\vec r}$ induces a rank-$\vec r$ PR of $G$ and vice 
versa. Hence the assertion follows.

\vspace{1ex}\noindent (ii)
For the chromatic number we also have $\chi(G^{\vec r},1)=\chi(G,\vec r)$, as 
it follows by an argument completely analogous to the proof of $d_\pi(G^{\vec 
r},1)=d_\pi(G,\vec r)$ (using colorings instead of projectors).
This implies,
 \begin{equation}
   \chi_f(G,\vec{r}) = \min_{b\in \naturals} \frac{\chi(G,b\vec{r})}{b} = \min_{b\in \naturals} \frac{\chi(G^{\vec{r}},b)}{b} = \chi_f(G^{\vec{r}}).
 \end{equation}

\vspace{1ex}\noindent (iii)
By definition, the weighted Lovász number of $\overline G$ is calculated as
\begin{equation}
  \vartheta(\overline G,\vec{r}) = \max_{\vec \rho,(\Pi_k)_k} \sum_{k\in V(\overline G)}r_k \tr(\rho \Pi_k),
\end{equation}
where the maximum is taken over all states $\rho$ and all PRs $(\Pi_k)_k$ of 
$\overline G$. However, if $(\Pi_k)_k$ is a PR of $\overline G$ then
$(\Pi_k)_{k,i}$ is a ($\vec r$-fold degenerate) PR of $\overline{G^{\vec r}}$, 
due to
\begin{equation}
E(\overline{G^{\vec{r}}})= \set{[(v,i),(w,j)]|[v,w]\in E(\overline G)}.
\end{equation}
 Thus, $\vartheta(\overline{G^{\vec{r}}}) \ge \vartheta(\overline 
G,\vec{r})$.
Conversely, let $(P_{k,i})_{k,i}$ be any PR of $\overline{G^{\vec r}}$.
For any state $\rho$ we let $P'_k = P_{k,\hat \imath}$ for $\hat 
\imath$ the index that maximizes $\tr(\rho P_{k,i})$.
Then $(P'_k)_k$ is a PR of $\overline G$ and hence 
$\vartheta(\overline G, \vec r)\ge \vartheta(\overline{G^{\vec r}})$.

\vspace{1ex}\noindent (iv)
This follows directly from the definition in Eq.~\eqref{eq:defchif} by 
substituting $x_I$ by $mx_I$ and $\vec r$ by $m\vec r$.

\vspace{1ex}\noindent (v)
This follows at once from the definition in Eq.~\eqref{eq:thTH}.

\section{Proof of Theorem~\ref{thm:subset}}
\label{app:subset}
The theorem consists of three statements: (i) $\overline{\rankb}(G)$ is convex, 
(ii) $\stabb(G) \subseteq \overline{\rankb}(G)$, and (iii) 
$\overline{\rankb}(G) \subseteq \thb(G)$.

\vspace{1ex} \noindent (i)
For any vector $\vec p\in \rankb(G)$ we can find a $d$-dimensional PR 
 $(\Pi_k)_k$ such that $p_k = \tr(\Pi_k)/d$. With $\Pi'_k$ and $d'$ accordingly 
 for $\vec p'\in \rankb(G)$, we let
\begin{equation}
\Gamma_k=
(\id_{d'}\otimes\Pi_k)\oplus
(\id_d\otimes\Pi'_k),
\end{equation}
where $A\oplus B$ denotes the block-diagonal matrix with blocks $A$ and $B$.
By construction, $(\Gamma_k)_k$ is a $(2dd')$-dimensional PR of $G$. Due to 
$\tr(\Gamma_k)/(2dd')=(d'\tr(\Pi_k)+d\tr(\Pi'_k))/(2dd')= (p_k+p'_k)/2$ we have 
$(\vec p+\vec p\,')/2\in \rankb(G)$.
Iterating this argument, any point $q\vec p+(1-q) \vec p'$ with $0\le q\le 1$ is 
arbitrarily close to some element of $\rankb(G)$ since any such $q$ can be 
arbitrarily well approximated by a fraction $x/2^n$ with $x,n\in \mathbb N$.
Hence $\overline{\rankb}(G)$ is convex.

\vspace{1ex} \noindent (ii)
Any extremal point $\vec a$ of $\stabb(G)$ is given by some independent set $I$ 
 of $G$ via $a_v=1$ if $v\in I$ and $a_v=0$ else.
Then $(a_v \id_d)_v$ is a $d$-dimensional PR with $\vec r=\vec ad$, that is, 
 $\vec a\in \rankb(G)$.
Since $\stabb(G)$ is the convex hull of its extremal points and $\rankb(G)$ is 
 convex, the assertion follows.

\vspace{1ex} \noindent (iii)
By definition, $\rankb(G)$ consists of all probability assignments involving 
the completely depolarized state and $\thb(G)$ consists of all probability 
assignments for any quantum state.
Since $\thb(G)$ is closed\cite{grotschel1986relaxations}, the assertion 
follows.

\section{Proof of Theorem~\ref{thm:stabisrank}}
\label{app:stabisrank}
The proof of Theorem~\ref{thm:stabisrank} is based on an exhaustive test of all
graphs with no more than $8$ vertices. Since the exact description of $\rankb(G)$
is difficult, we propose a linear relaxation of  $\rankb(G)$ by using the dimension
relations of union and intersection of subspaces.
Note that each rank-$r$ projector corresponds to a $r$-dimensional subspace.
More explicitly, for a given $d$-dimensional projector $\Pi$, denote $\Pi^s$ as 
the subspace spanned by all the vectors $\{\Pi \vec{v} | \forall \vec{v}\}$.  
Then we know that
\begin{eqnarray}\label{eq:relaxation}
 &\dim(\Pi^s) = \rank(\Pi) = \tr(\Pi) \leq d,\nonumber\\
 &\dim \Pi_1^s + \dim \Pi_2^s = \dim(\Pi_1^s+\Pi_2^s) + \dim(\Pi_1^s \cap 
 \Pi_2^s),\nonumber\\
 &\dim(\Pi_1^s \cap \Pi_2^s) \leq \min \{\dim \Pi_1^s, \dim \Pi_2^s\},
\end{eqnarray}
where $\Pi_1^s+\Pi_2^s = \{\vec{v}_1 + \vec{v}_2 |\forall \vec{v}_1 \in 
\Pi_1^s, \vec{v}_2 \in \Pi_2^s \}$ and $\Pi_1^s\cap\Pi_2^s = \{\vec{v} |\vec{v} 
\in \Pi_1^s \text{ and } \vec{v} \in \Pi_2^s \}$. To take more advantage of 
these relations, we consider the intersections of subspaces which are related 
to the projectors in the PR.  Denote $\Pi_I = \cap_{i\in 
I} \Pi_i$ for a given set $I$ of vertices in $G$ and let $\Pi_{\emptyset} = 
\id$. By definition, $\Pi_I = 0$ if  $I$ is not an independent set. This 
implies that  $\Pi_{I_1}$ and  $\Pi_{I_2}$ are orthogonal if  $I_1 \cup I_2$ is 
no longer an independent set for two given independent sets  $I_1, I_2$.

For a given graph $G$, denote the set of all independent sets as $\mathcal{I}$.  
Then define the corresponding independent set graph $\mathcal{G}$ as the graph 
such that
\begin{eqnarray}
  &V(\mathcal{G}) = \{v_I\}_{I\in \mathcal{I}},\ E(\mathcal{G}) = \{[v_{I_1}, 
v_{I_2}]|  \text{ if } I_1 \cup I_2 \not\in \mathcal{I}, I_1, I_2 \in 
\mathcal{I}\}.
\end{eqnarray}
For example, if $G = C_{5}$ is the $5$-cycle graph, then the independent set graph  
$\mathcal{G}$ is as shown in Fig.~\ref{fig:hypergh}.
 \begin{figure}[htpb]
  \centering
  \includegraphics[width=0.4\textwidth]{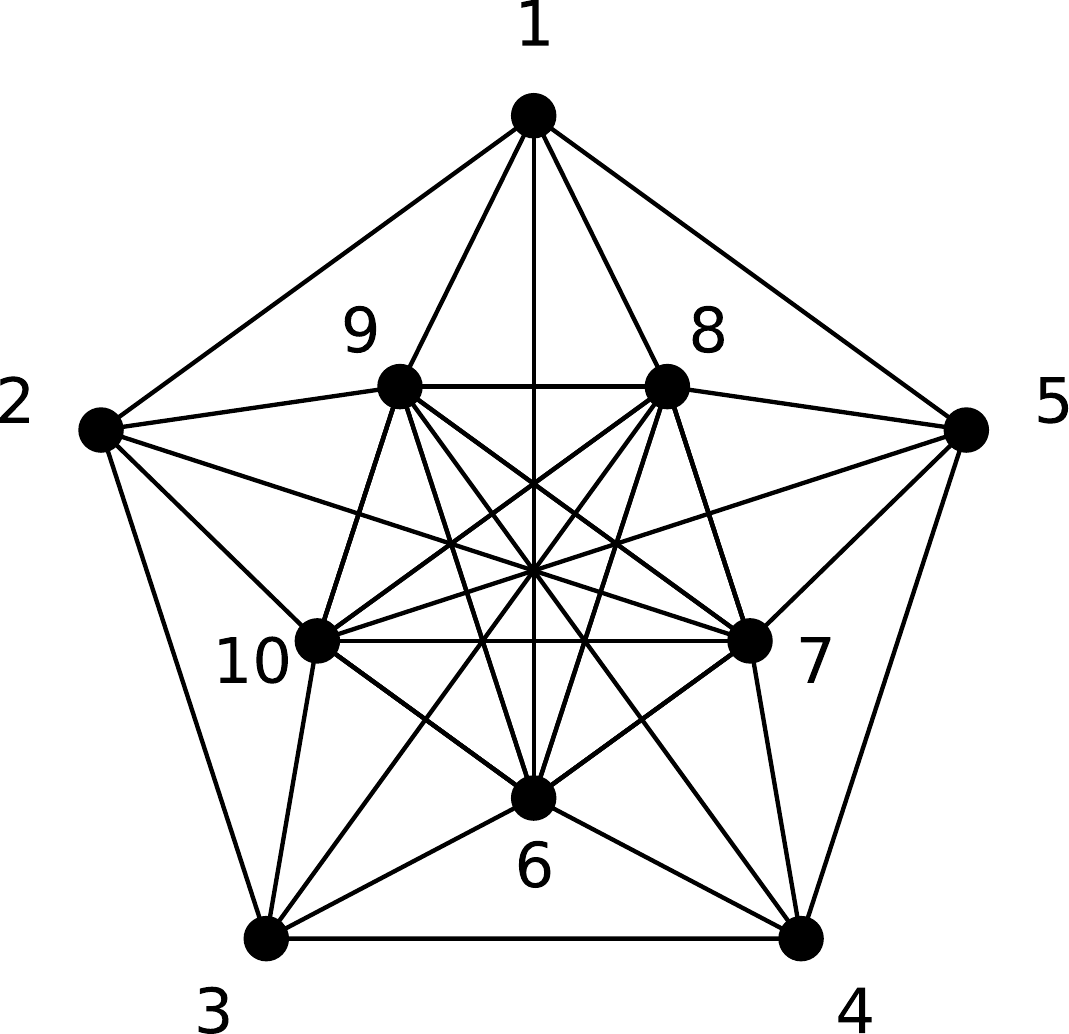}
  \caption{The independent set graph $\mathcal{G}$ for the  $5$-cycle graph 
  $C_5$, where the vertex $i$ represents the independent set  $\{i\} $ for 
$i=1,2,3,4,5$.  The vertices $6,7,8,9,10$ represent the independent sets 
$\{2,5\}, \{1,3\}, \{2,4\}, \{3,5\}, \{1,4\}$, respectively.}
  \label{fig:hypergh}
\end{figure}

Denote $\mathcal{C}$ as the set of all cliques in $\mathcal{G}$. For a given 
clique $C \in \mathcal{C}$, denote $H_C$ as the set of vertices in 
$V(\mathcal{G})$ which are connected to all vertices in $C$. That is,
\begin{equation}
  H_C := \{v_I|v_I\in V(\mathcal{G}), C \cup \{v_I\} \in \mathcal{C} \}.
\end{equation}
Then we have the following constraints on the PRs of $G$:
\begin{eqnarray}\label{eq:relaxation2}
  &\Pi_{I_1}^s \perp \Pi_{I_2}^s \text{ if } v_{I_1}, v_{I_2} \in C \Rightarrow 
  \sum_{v_I\in C}\overline{\dim}(\Pi^s_I) \leq 1, \forall C\in 
  \mathcal{C},\nonumber\\
  & \Pi_{I_1}^s + \Pi_{I_2}^s \subseteq \Pi_{I_1 \cap I_2}^s \Rightarrow 
  \sum_{i=1,2} \overline{\dim}(\Pi^s_{I_i}) \leq \overline{\dim}(\Pi^s_{I_1 
  \cap I_2}) + \overline{\dim}(\Pi^s_{I_1 \cup I_2}), \forall I_1, 
  I_2,\nonumber\\
  & \forall v_{I_1}, v_{I_2} \in H_C \Rightarrow \Pi^s_{I_1}+\Pi^s_{I_2} \perp 
  \sum_{v_I \in C} \Pi^s_I\nonumber\\ &\hspace{6em}\Rightarrow \sum_{v_I\in 
  C}\overline{\dim}(\Pi^s_I) + \sum_{i=1,2}\overline{\dim}(\Pi^s_{I_i}) \leq 1 
  + \overline{\dim}(\Pi^s_{I_1 \cup I_2}),
\end{eqnarray}
where $\overline{\dim}(\Pi) = \dim(\Pi)/d$.

By combining all the constraints in Eq.~\eqref{eq:relaxation2} with the 
non-negativity constraints, we have a polytope whose elements are possible 
values for $\{\overline{\dim}(\Pi_{I}^s)\}_{I \in \mathcal{I}} $. If we only 
consider the possible values of $\{\overline{\dim}(\Pi_{\{v_i\}})\}_{v_i \in 
V(G)}$, then we have a linear relaxation of $\rankb(G)$. We denote such a 
linear relaxation as $\lrankb(G)$. Note that we can add extra 
constraints that $\dim(\Pi_I) \in \naturals, \forall I\in \mathcal{I}$ if we 
only focus on a specific dimension $d$.

For a given graph, we can calculate $\lrankb(G)$ as described above with 
computer programs. If $\lrankb(G) = \stabb(G)$, then we know that $\rankb(G) = 
\stabb(G)$. As it turns out,  $\lrankb(G) = \stabb(G)$ if  $G$ is a graph with 
no more than  $8$ vertices. Thus, we have proved Theorem
\ref{thm:stabisrank}.

To have a closer look at this linear relaxation method, we illustrate it with 
odd cycles. It is known that $\stabb(G) = \thb(G)$ if $G$ is 
perfect~\cite{lovasz1979shannon}, which means that those graphs cannot be used 
to reveal quantum contextuality.
Recall that a graph is called perfect if all the induced subgraph of $G$ are 
not odd cycles or odd anti-cycles~\cite{chudnovsky2006strong}. Hence, odd 
cycles and odd anti-cycles are basic in the study of quantum 
contextuality~\cite{cabello2013basic}. Note that $\stabb(G)$ is a polytope 
which can be determined by the set of its facets $I(G,\vec{w}) = 
\alpha(G,\vec{w})$, where $\vec{w} \geq 0$. Each point outside of $\stabb(G)$ 
violates at least one of the tight inequalities, i.e., the inequalities 
defining the facets. For a given facet $I(G,\vec{w}) = \alpha(G,\vec{w})$, if 
the subgraph of $\{i|w_i>0\}$ is a clique, then we say that this facet is 
trivial.  This is because $\max I(G,\vec{w}) = 1$ in both the NCHV case and the 
quantum case. Thus, we only need to consider the non-trivial tight inequalities 
one by one. For the odd cycle $C_{2n+1}$ in Fig.~\ref{fig:multistate_eg}, the 
only non-trivial facet is~\cite{chudnovsky2006strong}
\begin{equation}
  \sum_{i=1}^{2n+1} P(\varepsilon_i) \leq n = \alpha(C_{2n+1}).
\end{equation}
If $\{\Pi_i\}_{i=1}^{2n+1}$ is a PR of the odd cycle $C_{2n+1}$, then 
Eq.\eqref{eq:relaxation2} implies that
\begin{eqnarray}\label{eq:oddcycle}
 & \overline{\dim}(\Pi^s_1) + \overline{\dim}(\Pi^s_2) + 
 \overline{\dim}(\Pi^s_{2n+1}) \leq 1 + 
 \overline{\dim}(\Pi^s_{\{2,2n+1\}}),\nonumber\\
 & \overline{\dim}(\Pi^s_{I_k}) + \overline{\dim}(\Pi^s_{k+1}) + 
 \overline{\dim}(\Pi^s_{2n-k}) \leq 1+ \overline{\dim}(\Pi^s_{I_{k+1}}),\ 
 \forall k = 1,\ldots,n-2,\nonumber\\
 & \overline{\dim}(\Pi^s_{I_{n-1}}) + \overline{\dim}(\Pi^s_{n+1}) + 
\overline{\dim}(\Pi^s_{n+2}) \leq 1,
\end{eqnarray}
where $I_k = \cup_{j=1}^{k}\{2j,2(n-j)+3\}$. Equation~\eqref{eq:oddcycle} 
implies that, for any PR $\{\Pi_i\}_{i=1}^{2n+1}$,
\begin{equation}\label{eq:cycleineq}
  \sum_{i=1}^{2n+1} \overline{\dim}(\Pi^s_i) \leq n.
\end{equation}
Thus, $\stabb(G) = \rankb(G)$ if $G$ is an odd cycle.

\begin{figure}[htpb]
  \centering
  \includegraphics[width=0.3\textwidth]{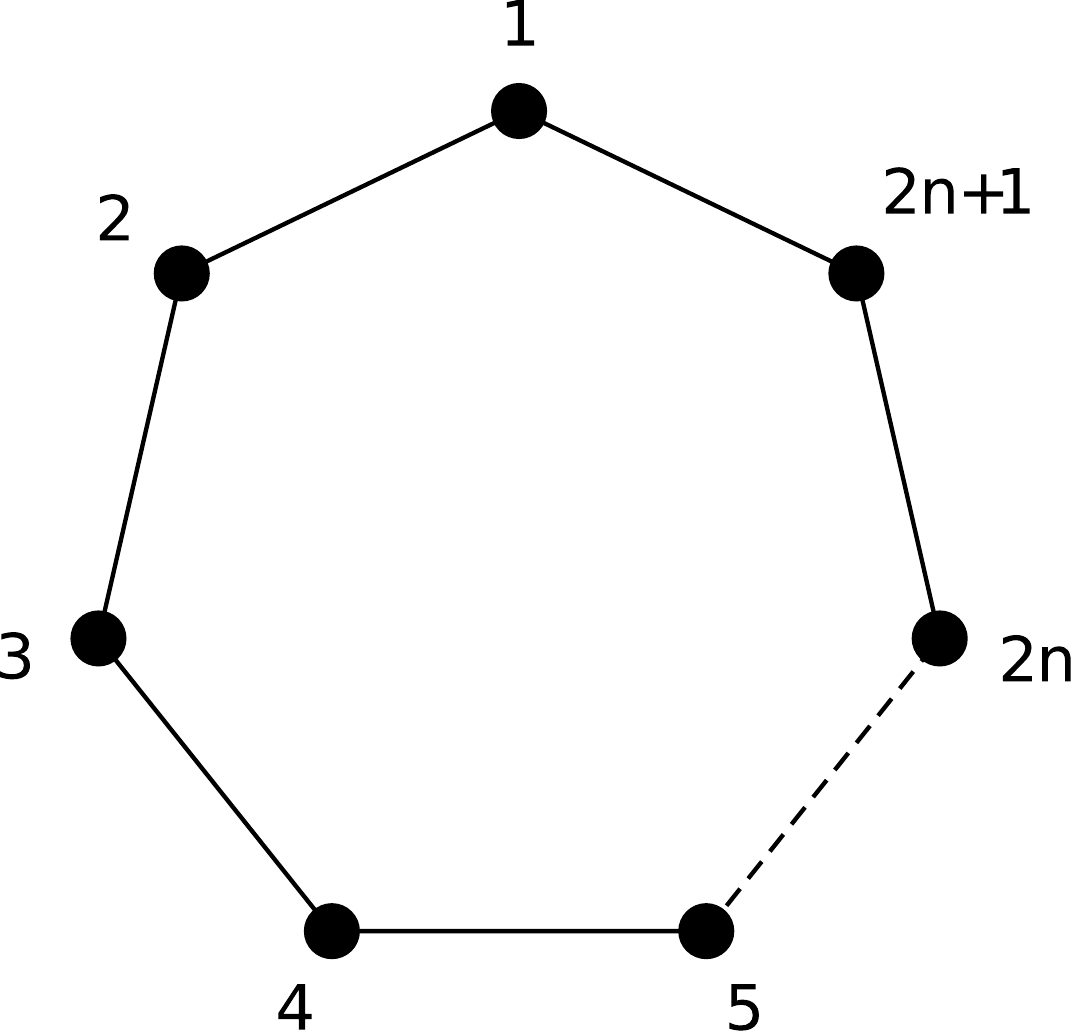}
  \caption{Odd cycle $C_{2n+1}$.}
  \label{fig:multistate_eg}
\end{figure}

\section{Implementation of the alternating optimization}
\label{app:seesaw}
Note that there exists a $(d\times n)$-matrix $Y$ such that $R=Y^\dagger Y$ if 
and only if $R\ge 0$ and $\rank(R)\le d$. Then, the fast implementation of the 
alternating optimization is based on the fact that the following two 
optimizations can be evaluated analytically:
\begin{equation}
  \begin{aligned}
    \min_{R}&\quad && \norm{R-X}_F\\
    \mathrm{s.t.}& && R\ge 0,~\rank(R)\le d,\\
  \end{aligned}
  \label{eq:optR}
\end{equation}
\begin{equation}
  \begin{aligned}
    \min_{L}&\quad && \norm{L-X}_F\\
    \mathrm{s.t.}& && L_{kk}=1,~L_{k\ell}=0~~\forall [k,\ell]\in E(G),\\
  \end{aligned}
  \label{eq:optL}
\end{equation}
where the Frobenius norm is defined as $\norm{M}_F=\tr(M^\dagger 
M)=\sum_{k\ell}\abs{M_{k\ell}}^2$.

The first optimization can be solved using a semidefinite variant of the 
Eckart--Young--Mirsky theorem \cite{eckart1936approximation}, which states that 
for any $n\times n$ matrix $M$, the best rank-$d$ (more precisely, rank no 
larger than $d$) approximation with respect the Frobenius norm (that is, 
$\min_{\rank(M_d)\le d}\norm{M_d-M}_F$) is achieved by
\begin{equation}
  M_d=U \diag(s_1,s_2,\dots,s_d,0,\dots,0) V^\dagger
  \label{eq:eckart}
\end{equation}
where $M=U\diag(s_1,s_2,\dots,s_n)V^\dagger$ is the singular value 
decomposition of $M$, and the singular values satisfy that $s_1\ge 
s_2\ge\dots\ge s_n\ge0$.
We mention that $M_d$ is not unique if $s_d$ is a degenerate singular value.
Now, let us consider the optimization in Eq.~\eqref{eq:optR}. As $X$ is 
Hermitian, it admits the decomposition $X=X^+-X^-$, where $X^+=P^+X P^+\ge0$, 
$X^-=-P^-X P^-\ge0$, and
\begin{equation}
  \begin{aligned}
    P^+&=\sum_{\lambda_k\ge 0}\ket{\varphi_k}\bra{\varphi_k},\\
    P^-&=\sum_{\lambda_k< 0}\ket{\varphi_k}\bra{\varphi_k}.
  \end{aligned}
  \label{eq:Xpm}
\end{equation}
Here $\lambda_1\ge\lambda_2\ge\dots\ge\lambda_n$ are the eigenvalues of $X$, 
and $\ket{\varphi_i}$ are the corresponding eigenvectors. Furthermore, let 
$R^+=P^+RP^+$, $R^-=P^-RP^-$, and
\begin{equation}
  X_d^+=\sum_{k\le d,\lambda_k\ge 0}\lambda_k\ket{\varphi_k}\bra{\varphi_k},
  \label{eq:optRvar}
\end{equation}
then the optimization in Eq.~\eqref{eq:optR} satisfies that
\begin{equation}
  \begin{aligned}
  \norm{R-X}_F&\ge\norm{R^++R^--X^++X^-}_F\\
  &=\norm{R^+-X^+}_F+\norm{R^-+X^-}_F\\
  &\ge\norm{X^+_d-X^+}_F+\norm{X^-}_F,
  \end{aligned}
  \label{eq:optRtri}
\end{equation}
where the first two lines follow from that 
$\norm{M}_F\ge\norm{P^+MP^++P^-MP^-}_F=\norm{P^+MP^+}_F+\norm{P^-MP^-}_F$, and 
the last line follows from the Eckart--Young--Mirsky theorem as well as the 
facts that $\rank(R^+)=\rank(P^+RP^+)\le\rank(R)\le d$ and 
$\norm{M_1+M_2}_F\ge\norm{M_1}_F$ when $M_1,M_2\ge 0$.
Moreover, one can easily verify that all inequalities in Eq.~\eqref{eq:optRtri} 
are saturated when $R=X_d^+$, because $P^+X_d^+P^+=X_d^+$ and $P^-X_d^+P^-=0$.  
By noting that $X_d^+$ satisfies that $X_d^+\ge 0$ and $\rank(X_d^+)\le d$, we 
get that the optimization in Eq.~\eqref{eq:optR} is achieved when $R=X_d^+$, 
which gives the solution
\begin{equation}
  \sum_{k\ge d+1,\lambda_k\ge 0}\lambda_k^2
  +\sum_{\lambda_k<0}\lambda_k^2.
  \label{eq:optRsol}
\end{equation}

The solution of the second optimization in Eq.~\eqref{eq:optL} follows directly 
from the definition of the Frobenius norm 
$\norm{M}_F=\sum_{k\ell}\abs{M_{k\ell}}^2$.  One can easily verify that the 
minimization is achieved when
\begin{equation}
  \begin{aligned}
    &L_{kk}=1,\quad &&k=1,2,\dots,n\\
    &L_{k\ell}=0,\quad &&[k,\ell]\in E(G),\\
    &L_{k\ell}=X_{k\ell},\quad &&k\ne\ell\text{ and }[k,\ell]\notin E(G),\\
  \end{aligned}
  \label{eq:optLvar}
\end{equation}
and the solution is
\begin{equation}
  \sum_{k=1}^d(1-X_{kk})^2+\sum_{[k,\ell]\in E(G)}\abs{X_{k\ell}}^2.
  \label{eq:optLsol}
\end{equation}

\bibliography{sic}

\begin{thebibliography}{35}%
\makeatletter
\providecommand \@ifxundefined [1]{%
 \@ifx{#1\undefined}
}%
\providecommand \@ifnum [1]{%
 \ifnum #1\expandafter \@firstoftwo
 \else \expandafter \@secondoftwo
 \fi
}%
\providecommand \@ifx [1]{%
 \ifx #1\expandafter \@firstoftwo
 \else \expandafter \@secondoftwo
 \fi
}%
\providecommand \natexlab [1]{#1}%
\providecommand \enquote  [1]{``#1''}%
\providecommand \bibnamefont  [1]{#1}%
\providecommand \bibfnamefont [1]{#1}%
\providecommand \citenamefont [1]{#1}%
\providecommand \href@noop [0]{\@secondoftwo}%
\providecommand \href [0]{\begingroup \@sanitize@url \@href}%
\providecommand \@href[1]{\@@startlink{#1}\@@href}%
\providecommand \@@href[1]{\endgroup#1\@@endlink}%
\providecommand \@sanitize@url [0]{\catcode `\\12\catcode `\$12\catcode
  `\&12\catcode `\#12\catcode `\^12\catcode `\_12\catcode `\%12\relax}%
\providecommand \@@startlink[1]{}%
\providecommand \@@endlink[0]{}%
\providecommand \url  [0]{\begingroup\@sanitize@url \@url }%
\providecommand \@url [1]{\endgroup\@href {#1}{\urlprefix }}%
\providecommand \urlprefix  [0]{URL }%
\providecommand \Eprint [0]{\href }%
\providecommand \doibase [0]{https://doi.org/}%
\providecommand \selectlanguage [0]{\@gobble}%
\providecommand \bibinfo  [0]{\@secondoftwo}%
\providecommand \bibfield  [0]{\@secondoftwo}%
\providecommand \translation [1]{[#1]}%
\providecommand \BibitemOpen [0]{}%
\providecommand \bibitemStop [0]{}%
\providecommand \bibitemNoStop [0]{.\EOS\space}%
\providecommand \EOS [0]{\spacefactor3000\relax}%
\providecommand \BibitemShut  [1]{\csname bibitem#1\endcsname}%
\let\auto@bib@innerbib\@empty
\bibitem [{\citenamefont {Xu}\ and\ \citenamefont
  {Cabello}(2019)}]{xu2019necessary}%
  \BibitemOpen
  \bibfield  {author} {\bibinfo {author} {\bibfnamefont {Z.-P.}\ \bibnamefont
  {Xu}}\ and\ \bibinfo {author} {\bibfnamefont {A.}~\bibnamefont {Cabello}},\
  }\bibfield  {title} {\enquote {\bibinfo {title} {Necessary and sufficient
  condition for contextuality from incompatibility},}\ }\href
  {https://doi.org/10.1103/PhysRevA.99.020103} {\bibfield  {journal} {\bibinfo
  {journal} {Phys. Rev. A}\ }\textbf {\bibinfo {volume} {99}},\ \bibinfo
  {pages} {020103} (\bibinfo {year} {2019})}\BibitemShut {NoStop}%
\bibitem [{\citenamefont {Bell}(1964)}]{bell1964epr}%
  \BibitemOpen
  \bibfield  {author} {\bibinfo {author} {\bibfnamefont {J.~S.}\ \bibnamefont
  {Bell}},\ }\bibfield  {title} {\enquote {\bibinfo {title} {On the {Einstein}
  {Podolsky} {Rosen} paradox},}\ }\href
  {https://doi.org/10.1103/PhysicsPhysiqueFizika.1.195} {\bibfield  {journal}
  {\bibinfo  {journal} {Physics}\ }\textbf {\bibinfo {volume} {1}},\ \bibinfo
  {pages} {195} (\bibinfo {year} {1964})}\BibitemShut {NoStop}%
\bibitem [{\citenamefont {Cubitt}\ \emph {et~al.}(2010)\citenamefont {Cubitt},
  \citenamefont {Leung}, \citenamefont {Matthews},\ and\ \citenamefont
  {Winter}}]{cubitt2010improving}%
  \BibitemOpen
  \bibfield  {author} {\bibinfo {author} {\bibfnamefont {T.~S.}\ \bibnamefont
  {Cubitt}}, \bibinfo {author} {\bibfnamefont {D.}~\bibnamefont {Leung}},
  \bibinfo {author} {\bibfnamefont {W.}~\bibnamefont {Matthews}},\ and\
  \bibinfo {author} {\bibfnamefont {A.}~\bibnamefont {Winter}},\ }\bibfield
  {title} {\enquote {\bibinfo {title} {Improving zero-error classical
  communication with entanglement},}\ }\href
  {https://doi.org/10.1103/PhysRevLett.104.230503} {\bibfield  {journal}
  {\bibinfo  {journal} {Phys. Rev. Lett.}\ }\textbf {\bibinfo {volume} {104}},\
  \bibinfo {pages} {230503} (\bibinfo {year} {2010})}\BibitemShut {NoStop}%
\bibitem [{\citenamefont {Saha}, \citenamefont {Horodecki},\ and\ \citenamefont
  {Paw{\l}owski}(2019)}]{saha2019state}%
  \BibitemOpen
  \bibfield  {author} {\bibinfo {author} {\bibfnamefont {D.}~\bibnamefont
  {Saha}}, \bibinfo {author} {\bibfnamefont {P.}~\bibnamefont {Horodecki}},\
  and\ \bibinfo {author} {\bibfnamefont {M.}~\bibnamefont {Paw{\l}owski}},\
  }\bibfield  {title} {\enquote {\bibinfo {title} {State independent
  contextuality advances one-way communication},}\ }\href
  {https://doi.org/10.1088/1367-2630/ab4149} {\bibfield  {journal} {\bibinfo
  {journal} {New J. Phys.}\ }\textbf {\bibinfo {volume} {21}},\ \bibinfo
  {pages} {093057} (\bibinfo {year} {2019})}\BibitemShut {NoStop}%
\bibitem [{\citenamefont {Cabello}\ \emph {et~al.}(2011)\citenamefont
  {Cabello}, \citenamefont {D'Ambrosio}, \citenamefont {Nagali},\ and\
  \citenamefont {Sciarrino}}]{cabello2011hybrid}%
  \BibitemOpen
  \bibfield  {author} {\bibinfo {author} {\bibfnamefont {A.}~\bibnamefont
  {Cabello}}, \bibinfo {author} {\bibfnamefont {V.}~\bibnamefont {D'Ambrosio}},
  \bibinfo {author} {\bibfnamefont {E.}~\bibnamefont {Nagali}},\ and\ \bibinfo
  {author} {\bibfnamefont {F.}~\bibnamefont {Sciarrino}},\ }\bibfield  {title}
  {\enquote {\bibinfo {title} {Hybrid ququart-encoded quantum cryptography
  protected by {Kochen}--{Specker} contextuality},}\ }\href
  {https://doi.org/10.1103/PhysRevA.84.030302} {\bibfield  {journal} {\bibinfo
  {journal} {Phys. Rev. A}\ }\textbf {\bibinfo {volume} {84}},\ \bibinfo
  {pages} {030302} (\bibinfo {year} {2011})}\BibitemShut {NoStop}%
\bibitem [{\citenamefont {Ekert}(1991)}]{ekert1991quantum}%
  \BibitemOpen
  \bibfield  {author} {\bibinfo {author} {\bibfnamefont {A.~K.}\ \bibnamefont
  {Ekert}},\ }\bibfield  {title} {\enquote {\bibinfo {title} {Quantum
  cryptography based on {Bell}'s theorem},}\ }\href
  {https://doi.org/10.1103/PhysRevLett.67.661} {\bibfield  {journal} {\bibinfo
  {journal} {Phys. Rev. Lett.}\ }\textbf {\bibinfo {volume} {67}},\ \bibinfo
  {pages} {661} (\bibinfo {year} {1991})}\BibitemShut {NoStop}%
\bibitem [{\citenamefont {Howard}\ \emph {et~al.}(2014)\citenamefont {Howard},
  \citenamefont {Wallman}, \citenamefont {Veitch},\ and\ \citenamefont
  {Emerson}}]{howard2014contextuality}%
  \BibitemOpen
  \bibfield  {author} {\bibinfo {author} {\bibfnamefont {M.}~\bibnamefont
  {Howard}}, \bibinfo {author} {\bibfnamefont {J.}~\bibnamefont {Wallman}},
  \bibinfo {author} {\bibfnamefont {V.}~\bibnamefont {Veitch}},\ and\ \bibinfo
  {author} {\bibfnamefont {J.}~\bibnamefont {Emerson}},\ }\bibfield  {title}
  {\enquote {\bibinfo {title} {Contextuality supplies the `magic' for quantum
  computation},}\ }\href {https://doi.org/10.1038/nature13460} {\bibfield
  {journal} {\bibinfo  {journal} {Nature (London)}\ }\textbf {\bibinfo {volume}
  {510}},\ \bibinfo {pages} {351} (\bibinfo {year} {2014})}\BibitemShut
  {NoStop}%
\bibitem [{\citenamefont {Raussendorf}(2013)}]{raussendorf2013contextuality}%
  \BibitemOpen
  \bibfield  {author} {\bibinfo {author} {\bibfnamefont {R.}~\bibnamefont
  {Raussendorf}},\ }\bibfield  {title} {\enquote {\bibinfo {title}
  {Contextuality in measurement-based quantum computation},}\ }\href
  {https://doi.org/10.1103/PhysRevA.88.022322} {\bibfield  {journal} {\bibinfo
  {journal} {Phys. Rev. A}\ }\textbf {\bibinfo {volume} {88}},\ \bibinfo
  {pages} {022322} (\bibinfo {year} {2013})}\BibitemShut {NoStop}%
\bibitem [{\citenamefont {Barrett}, \citenamefont {Hardy},\ and\ \citenamefont
  {Kent}(2005)}]{barrett2005no}%
  \BibitemOpen
  \bibfield  {author} {\bibinfo {author} {\bibfnamefont {J.}~\bibnamefont
  {Barrett}}, \bibinfo {author} {\bibfnamefont {L.}~\bibnamefont {Hardy}},\
  and\ \bibinfo {author} {\bibfnamefont {A.}~\bibnamefont {Kent}},\ }\bibfield
  {title} {\enquote {\bibinfo {title} {No signaling and quantum key
  distribution},}\ }\href {https://doi.org/10.1103/PhysRevLett.95.010503}
  {\bibfield  {journal} {\bibinfo  {journal} {Phys. Rev. Lett.}\ }\textbf
  {\bibinfo {volume} {95}},\ \bibinfo {pages} {010503} (\bibinfo {year}
  {2005})}\BibitemShut {NoStop}%
\bibitem [{\citenamefont {Kochen}\ and\ \citenamefont
  {Specker}(1968)}]{kochen1968problem}%
  \BibitemOpen
  \bibfield  {author} {\bibinfo {author} {\bibfnamefont {S.}~\bibnamefont
  {Kochen}}\ and\ \bibinfo {author} {\bibfnamefont {E.~P.}\ \bibnamefont
  {Specker}},\ }\bibfield  {title} {\enquote {\bibinfo {title} {The problem of
  hidden variables in quantum mechanics},}\ }\href
  {https://doi.org/10.1512/iumj.1968.17.17004} {\bibfield  {journal} {\bibinfo
  {journal} {Indiana Univ. Math. J.}\ }\textbf {\bibinfo {volume} {17}},\
  \bibinfo {pages} {59} (\bibinfo {year} {1968})}\BibitemShut {NoStop}%
\bibitem [{\citenamefont {Cabello}, \citenamefont {Estebaranz},\ and\
  \citenamefont {{Garc{\'i}a-Alcaine}}(1996)}]{cabello1996bell}%
  \BibitemOpen
  \bibfield  {author} {\bibinfo {author} {\bibfnamefont {A.}~\bibnamefont
  {Cabello}}, \bibinfo {author} {\bibfnamefont {J.}~\bibnamefont
  {Estebaranz}},\ and\ \bibinfo {author} {\bibfnamefont {G.}~\bibnamefont
  {{Garc{\'i}a-Alcaine}}},\ }\bibfield  {title} {\enquote {\bibinfo {title}
  {{Bell}--{Kochen}--{Specker} theorem: A proof with 18 vectors},}\ }\href
  {https://doi.org/10.1016/0375-9601(96)00134-X} {\bibfield  {journal}
  {\bibinfo  {journal} {Phys. Lett. A}\ }\textbf {\bibinfo {volume} {212}},\
  \bibinfo {pages} {183} (\bibinfo {year} {1996})}\BibitemShut {NoStop}%
\bibitem [{\citenamefont {Yu}\ and\ \citenamefont
  {Oh}(2012)}]{yu2012stateindependent}%
  \BibitemOpen
  \bibfield  {author} {\bibinfo {author} {\bibfnamefont {S.}~\bibnamefont
  {Yu}}\ and\ \bibinfo {author} {\bibfnamefont {C.~H.}\ \bibnamefont {Oh}},\
  }\bibfield  {title} {\enquote {\bibinfo {title} {State-independent proof of
  {Kochen}--{Specker} theorem with 13 rays},}\ }\href
  {https://doi.org/10.1103/PhysRevLett.108.030402} {\bibfield  {journal}
  {\bibinfo  {journal} {Phys. Rev. Lett.}\ }\textbf {\bibinfo {volume} {108}},\
  \bibinfo {pages} {030402} (\bibinfo {year} {2012})}\BibitemShut {NoStop}%
\bibitem [{\citenamefont {Cabello}, \citenamefont {Kleinmann},\ and\
  \citenamefont {Portillo}(2016)}]{cabello2016quantum}%
  \BibitemOpen
  \bibfield  {author} {\bibinfo {author} {\bibfnamefont {A.}~\bibnamefont
  {Cabello}}, \bibinfo {author} {\bibfnamefont {M.}~\bibnamefont {Kleinmann}},\
  and\ \bibinfo {author} {\bibfnamefont {J.~R.}\ \bibnamefont {Portillo}},\
  }\bibfield  {title} {\enquote {\bibinfo {title} {Quantum state-independent
  contextuality requires 13 rays},}\ }\href
  {https://doi.org/10.1088/1751-8113/49/38/38LT01} {\bibfield  {journal}
  {\bibinfo  {journal} {J. Phys. A: Math. Theor.}\ }\textbf {\bibinfo {volume}
  {49}},\ \bibinfo {pages} {38LT01} (\bibinfo {year} {2016})}\BibitemShut
  {NoStop}%
\bibitem [{\citenamefont {Peres}(1990)}]{peres1990incompatible}%
  \BibitemOpen
  \bibfield  {author} {\bibinfo {author} {\bibfnamefont {A.}~\bibnamefont
  {Peres}},\ }\bibfield  {title} {\enquote {\bibinfo {title} {Incompatible
  results of quantum measurements},}\ }\href@noop {} {\bibfield  {journal}
  {\bibinfo  {journal} {Physics Letters A}\ }\textbf {\bibinfo {volume}
  {151}},\ \bibinfo {pages} {107} (\bibinfo {year} {1990})}\BibitemShut
  {NoStop}%
\bibitem [{\citenamefont {Mermin}(1990)}]{mermin1990simple}%
  \BibitemOpen
  \bibfield  {author} {\bibinfo {author} {\bibfnamefont {N.~D.}\ \bibnamefont
  {Mermin}},\ }\bibfield  {title} {\enquote {\bibinfo {title} {Simple unified
  form for the major no-hidden-variables theorems},}\ }\href
  {https://doi.org/10.1103/PhysRevLett.65.3373} {\bibfield  {journal} {\bibinfo
   {journal} {Phys. Rev. Lett.}\ }\textbf {\bibinfo {volume} {65}},\ \bibinfo
  {pages} {3373} (\bibinfo {year} {1990})}\BibitemShut {NoStop}%
\bibitem [{\citenamefont {Kernaghan}\ and\ \citenamefont
  {Peres}(1995)}]{kernaghan1995kochen}%
  \BibitemOpen
  \bibfield  {author} {\bibinfo {author} {\bibfnamefont {M.}~\bibnamefont
  {Kernaghan}}\ and\ \bibinfo {author} {\bibfnamefont {A.}~\bibnamefont
  {Peres}},\ }\bibfield  {title} {\enquote {\bibinfo {title}
  {{Kochen}-{Specker} theorem for eight-dimensional space},}\ }\href
  {https://doi.org/10.1016/0375-9601(95)00012-R} {\bibfield  {journal}
  {\bibinfo  {journal} {Phys. Lett. A}\ }\textbf {\bibinfo {volume} {198}},\
  \bibinfo {pages} {1} (\bibinfo {year} {1995})}\BibitemShut {NoStop}%
\bibitem [{\citenamefont {Toh}(2013{\natexlab{a}})}]{toh2013kochen}%
  \BibitemOpen
  \bibfield  {author} {\bibinfo {author} {\bibfnamefont {S.~P.}\ \bibnamefont
  {Toh}},\ }\bibfield  {title} {\enquote {\bibinfo {title} {Kochen--{Specker}
  sets with a mixture of 16 rank-1 and 14 rank-2 projectors for a three-qubit
  system},}\ }\href {https://doi.org/10.1088/0256-307X/30/10/100302} {\bibfield
   {journal} {\bibinfo  {journal} {Chinese Phys. Lett.}\ }\textbf {\bibinfo
  {volume} {30}},\ \bibinfo {pages} {100302} (\bibinfo {year}
  {2013}{\natexlab{a}})}\BibitemShut {NoStop}%
\bibitem [{\citenamefont {Toh}(2013{\natexlab{b}})}]{toh2013stateindependent}%
  \BibitemOpen
  \bibfield  {author} {\bibinfo {author} {\bibfnamefont {S.~P.}\ \bibnamefont
  {Toh}},\ }\bibfield  {title} {\enquote {\bibinfo {title} {State-independent
  proof of {Kochen}--{Specker} theorem with thirty rank-two projectors},}\
  }\href {https://doi.org/10.1088/0256-307X/30/10/100303} {\bibfield  {journal}
  {\bibinfo  {journal} {Chinese Phys. Lett.}\ }\textbf {\bibinfo {volume}
  {30}},\ \bibinfo {pages} {100303} (\bibinfo {year}
  {2013}{\natexlab{b}})}\BibitemShut {NoStop}%
\bibitem [{\citenamefont {Mermin}(1993)}]{mermin1993hidden}%
  \BibitemOpen
  \bibfield  {author} {\bibinfo {author} {\bibfnamefont {N.~D.}\ \bibnamefont
  {Mermin}},\ }\bibfield  {title} {\enquote {\bibinfo {title} {Hidden variables
  and the two theorems of john bell},}\ }\href
  {https://doi.org/10.1103/RevModPhys.65.803} {\bibfield  {journal} {\bibinfo
  {journal} {Rev. Mod. Phys.}\ }\textbf {\bibinfo {volume} {65}},\ \bibinfo
  {pages} {803} (\bibinfo {year} {1993})}\BibitemShut {NoStop}%
\bibitem [{\citenamefont {Ramanathan}\ and\ \citenamefont
  {Horodecki}(2014)}]{ramanathan2014necessary}%
  \BibitemOpen
  \bibfield  {author} {\bibinfo {author} {\bibfnamefont {R.}~\bibnamefont
  {Ramanathan}}\ and\ \bibinfo {author} {\bibfnamefont {P.}~\bibnamefont
  {Horodecki}},\ }\bibfield  {title} {\enquote {\bibinfo {title} {Necessary and
  sufficient condition for state-independent contextual measurement
  scenarios},}\ }\href {https://doi.org/10.1103/PhysRevLett.112.040404}
  {\bibfield  {journal} {\bibinfo  {journal} {Phys. Rev. Lett.}\ }\textbf
  {\bibinfo {volume} {112}},\ \bibinfo {pages} {040404} (\bibinfo {year}
  {2014})}\BibitemShut {NoStop}%
\bibitem [{\citenamefont {Cabello}, \citenamefont {Severini},\ and\
  \citenamefont {Winter}(2014)}]{cabello2014graphtheoretic}%
  \BibitemOpen
  \bibfield  {author} {\bibinfo {author} {\bibfnamefont {A.}~\bibnamefont
  {Cabello}}, \bibinfo {author} {\bibfnamefont {S.}~\bibnamefont {Severini}},\
  and\ \bibinfo {author} {\bibfnamefont {A.}~\bibnamefont {Winter}},\
  }\bibfield  {title} {\enquote {\bibinfo {title} {Graph-theoretic approach to
  quantum correlations},}\ }\href
  {https://doi.org/10.1103/PhysRevLett.112.040401} {\bibfield  {journal}
  {\bibinfo  {journal} {Phys. Rev. Lett.}\ }\textbf {\bibinfo {volume} {112}},\
  \bibinfo {pages} {040401} (\bibinfo {year} {2014})}\BibitemShut {NoStop}%
\bibitem [{\citenamefont {Gr{\"o}tschel}, \citenamefont {Lov{\'a}sz},\ and\
  \citenamefont {Schrijver}(1984)}]{grotschel1984polynomial}%
  \BibitemOpen
  \bibfield  {author} {\bibinfo {author} {\bibfnamefont {M.}~\bibnamefont
  {Gr{\"o}tschel}}, \bibinfo {author} {\bibfnamefont {L.}~\bibnamefont
  {Lov{\'a}sz}},\ and\ \bibinfo {author} {\bibfnamefont {A.}~\bibnamefont
  {Schrijver}},\ }\bibfield  {title} {\enquote {\bibinfo {title} {Polynomial
  algorithms for perfect graphs},}\ }\href
  {https://doi.org/10.1016/S0304-0208(08)72943-8} {\bibfield  {journal}
  {\bibinfo  {journal} {Ann. Discrete. Math.}\ }\textbf {\bibinfo {volume}
  {21}},\ \bibinfo {pages} {325} (\bibinfo {year} {1984})}\BibitemShut
  {NoStop}%
\bibitem [{\citenamefont {Lovasz}(1979)}]{lovasz1979shannon}%
  \BibitemOpen
  \bibfield  {author} {\bibinfo {author} {\bibfnamefont {L.}~\bibnamefont
  {Lovasz}},\ }\bibfield  {title} {\enquote {\bibinfo {title} {On the {Shannon}
  capacity of a graph},}\ }\href {https://doi.org/10.1109/TIT.1979.1055985}
  {\bibfield  {journal} {\bibinfo  {journal} {IEEE Trans. Inf. Theory}\
  }\textbf {\bibinfo {volume} {25}},\ \bibinfo {pages} {1} (\bibinfo {year}
  {1979})}\BibitemShut {NoStop}%
\bibitem [{\citenamefont {Cabello}, \citenamefont {Kleinmann},\ and\
  \citenamefont {Budroni}(2015)}]{cabello2015necessary}%
  \BibitemOpen
  \bibfield  {author} {\bibinfo {author} {\bibfnamefont {A.}~\bibnamefont
  {Cabello}}, \bibinfo {author} {\bibfnamefont {M.}~\bibnamefont {Kleinmann}},\
  and\ \bibinfo {author} {\bibfnamefont {C.}~\bibnamefont {Budroni}},\
  }\bibfield  {title} {\enquote {\bibinfo {title} {Necessary and sufficient
  condition for quantum state-independent contextuality},}\ }\href
  {https://doi.org/10.1103/PhysRevLett.114.250402} {\bibfield  {journal}
  {\bibinfo  {journal} {Phys. Rev. Lett.}\ }\textbf {\bibinfo {volume} {114}},\
  \bibinfo {pages} {250402} (\bibinfo {year} {2015})}\BibitemShut {NoStop}%
\bibitem [{\citenamefont {Schrijver}(2004)}]{schrijver2004fractional}%
  \BibitemOpen
  \bibfield  {author} {\bibinfo {author} {\bibfnamefont {A.}~\bibnamefont
  {Schrijver}},\ }\enquote {\bibinfo {title} {Fractional and weighted colouring
  numbers},}\ in\ \href@noop {} {\emph {\bibinfo {booktitle} {Combinatorial
  Optimization}}}\ (\bibinfo  {publisher} {Springer-Verlag},\ \bibinfo
  {address} {Berlin},\ \bibinfo {year} {2004})\ p.\ \bibinfo {pages}
  {1096}\BibitemShut {NoStop}%
\bibitem [{\citenamefont {Man{\v c}inska}\ and\ \citenamefont
  {Roberson}(2016)}]{mancinska2016quantum}%
  \BibitemOpen
  \bibfield  {author} {\bibinfo {author} {\bibfnamefont {L.}~\bibnamefont
  {Man{\v c}inska}}\ and\ \bibinfo {author} {\bibfnamefont {D.~E.}\
  \bibnamefont {Roberson}},\ }\bibfield  {title} {\enquote {\bibinfo {title}
  {Quantum homomorphisms},}\ }\href
  {https://doi.org/10.1016/j.jctb.2015.12.009} {\bibfield  {journal} {\bibinfo
  {journal} {Journal of Combinatorial Theory, Series B}\ }\textbf {\bibinfo
  {volume} {118}},\ \bibinfo {pages} {228} (\bibinfo {year}
  {2016})}\BibitemShut {NoStop}%
\bibitem [{\citenamefont {McKay}\ and\ \citenamefont
  {Piperno}(2014)}]{mckay2014practical}%
  \BibitemOpen
  \bibfield  {author} {\bibinfo {author} {\bibfnamefont {B.~D.}\ \bibnamefont
  {McKay}}\ and\ \bibinfo {author} {\bibfnamefont {A.}~\bibnamefont
  {Piperno}},\ }\bibfield  {title} {\enquote {\bibinfo {title} {Practical graph
  isomorphism, {II}},}\ }\href {https://doi.org/10.1016/j.jsc.2013.09.003}
  {\bibfield  {journal} {\bibinfo  {journal} {J. Symb. Comput.}\ }\textbf
  {\bibinfo {volume} {60}},\ \bibinfo {pages} {94} (\bibinfo {year}
  {2014})}\BibitemShut {NoStop}%
\bibitem [{\citenamefont {Press}\ \emph {et~al.}(2007)\citenamefont {Press},
  \citenamefont {Teukolsky}, \citenamefont {Vetterling},\ and\ \citenamefont
  {Flannery}}]{press2007numerical}%
  \BibitemOpen
  \bibfield  {author} {\bibinfo {author} {\bibfnamefont {W.~H.}\ \bibnamefont
  {Press}}, \bibinfo {author} {\bibfnamefont {S.~A.}\ \bibnamefont
  {Teukolsky}}, \bibinfo {author} {\bibfnamefont {W.~T.}\ \bibnamefont
  {Vetterling}},\ and\ \bibinfo {author} {\bibfnamefont {B.~P.}\ \bibnamefont
  {Flannery}},\ }\href@noop {} {\emph {\bibinfo {title} {Numerical Recipes}}}\
  (\bibinfo  {publisher} {Cambridge University Press},\ \bibinfo {address}
  {Cambridge},\ \bibinfo {year} {2007})\BibitemShut {NoStop}%
\bibitem [{\citenamefont {Klyachko}\ \emph {et~al.}(2008)\citenamefont
  {Klyachko}, \citenamefont {Can}, \citenamefont {Binicio{\u g}lu},\ and\
  \citenamefont {Shumovsky}}]{klyachko2008simple}%
  \BibitemOpen
  \bibfield  {author} {\bibinfo {author} {\bibfnamefont {A.~A.}\ \bibnamefont
  {Klyachko}}, \bibinfo {author} {\bibfnamefont {M.~A.}\ \bibnamefont {Can}},
  \bibinfo {author} {\bibfnamefont {S.}~\bibnamefont {Binicio{\u g}lu}},\ and\
  \bibinfo {author} {\bibfnamefont {A.~S.}\ \bibnamefont {Shumovsky}},\
  }\bibfield  {title} {\enquote {\bibinfo {title} {Simple test for hidden
  variables in spin-1 systems},}\ }\href
  {https://doi.org/10.1103/PhysRevLett.101.020403} {\bibfield  {journal}
  {\bibinfo  {journal} {Phys. Rev. Lett.}\ }\textbf {\bibinfo {volume} {101}},\
  \bibinfo {pages} {020403} (\bibinfo {year} {2008})}\BibitemShut {NoStop}%
\bibitem [{\citenamefont {Gr\"otschel}, \citenamefont {Lov\'asz},\ and\
  \citenamefont {Schrijver}(1986)}]{grotschel1986relaxations}%
  \BibitemOpen
  \bibfield  {author} {\bibinfo {author} {\bibfnamefont {M.}~\bibnamefont
  {Gr\"otschel}}, \bibinfo {author} {\bibfnamefont {L.}~\bibnamefont
  {Lov\'asz}},\ and\ \bibinfo {author} {\bibfnamefont {A.}~\bibnamefont
  {Schrijver}},\ }\bibfield  {title} {\enquote {\bibinfo {title} {Relaxations
  of vertex packing},}\ }\href {https://doi.org/10.1007/BF02579273} {\bibfield
  {journal} {\bibinfo  {journal} {J. Comb. Theor.}\ }\textbf {\bibinfo {volume}
  {40}},\ \bibinfo {pages} {330} (\bibinfo {year} {1986})}\BibitemShut
  {NoStop}%
\bibitem [{\citenamefont {Knuth}(1994)}]{knuth1994sandwich}%
  \BibitemOpen
  \bibfield  {author} {\bibinfo {author} {\bibfnamefont {D.~E.}\ \bibnamefont
  {Knuth}},\ }\bibfield  {title} {\enquote {\bibinfo {title} {The sandwich
  theorem},}\ }\href {https://doi.org/10.37236/1193} {\bibfield  {journal}
  {\bibinfo  {journal} {Electron. J. Comb.}\ }\textbf {\bibinfo {volume} {1}},\
  \bibinfo {pages} {A1} (\bibinfo {year} {1994})}\BibitemShut {NoStop}%
\bibitem [{\citenamefont {Scheinerman}\ and\ \citenamefont
  {Ullman}(1997)}]{scheinerman1997fractional}%
  \BibitemOpen
  \bibfield  {author} {\bibinfo {author} {\bibfnamefont {E.~R.}\ \bibnamefont
  {Scheinerman}}\ and\ \bibinfo {author} {\bibfnamefont {D.~H.}\ \bibnamefont
  {Ullman}},\ }\href@noop {} {\emph {\bibinfo {title} {Fractional Graph Theory.
  A Rational Approach to the Theory of Graphs}}}\ (\bibinfo  {publisher}
  {Wiley},\ \bibinfo {address} {New York},\ \bibinfo {year} {1997})\BibitemShut
  {NoStop}%
\bibitem [{\citenamefont {Chudnovsky}\ \emph {et~al.}(2006)\citenamefont
  {Chudnovsky}, \citenamefont {Robertson}, \citenamefont {Seymour},\ and\
  \citenamefont {Thomas}}]{chudnovsky2006strong}%
  \BibitemOpen
  \bibfield  {author} {\bibinfo {author} {\bibfnamefont {M.}~\bibnamefont
  {Chudnovsky}}, \bibinfo {author} {\bibfnamefont {N.}~\bibnamefont
  {Robertson}}, \bibinfo {author} {\bibfnamefont {P.}~\bibnamefont {Seymour}},\
  and\ \bibinfo {author} {\bibfnamefont {R.}~\bibnamefont {Thomas}},\
  }\bibfield  {title} {\enquote {\bibinfo {title} {The strong perfect graph
  theorem},}\ }\href@noop {} {\bibfield  {journal} {\bibinfo  {journal} {Ann.
  Math.}\ ,\ \bibinfo {pages} {51}} (\bibinfo {year} {2006})}\BibitemShut
  {NoStop}%
\bibitem [{\citenamefont {Cabello}\ \emph {et~al.}(2013)\citenamefont
  {Cabello}, \citenamefont {Danielsen}, \citenamefont {{L{\'o}pez-Tarrida}},\
  and\ \citenamefont {Portillo}}]{cabello2013basic}%
  \BibitemOpen
  \bibfield  {author} {\bibinfo {author} {\bibfnamefont {A.}~\bibnamefont
  {Cabello}}, \bibinfo {author} {\bibfnamefont {L.~E.}\ \bibnamefont
  {Danielsen}}, \bibinfo {author} {\bibfnamefont {A.~J.}\ \bibnamefont
  {{L{\'o}pez-Tarrida}}},\ and\ \bibinfo {author} {\bibfnamefont {J.~R.}\
  \bibnamefont {Portillo}},\ }\bibfield  {title} {\enquote {\bibinfo {title}
  {Basic exclusivity graphs in quantum correlations},}\ }\href
  {https://doi.org/10.1103/PhysRevA.88.032104} {\bibfield  {journal} {\bibinfo
  {journal} {Phys. Rev. A}\ }\textbf {\bibinfo {volume} {88}},\ \bibinfo
  {pages} {032104} (\bibinfo {year} {2013})}\BibitemShut {NoStop}%
\bibitem [{\citenamefont {Eckart}\ and\ \citenamefont
  {Young}(1936)}]{eckart1936approximation}%
  \BibitemOpen
  \bibfield  {author} {\bibinfo {author} {\bibfnamefont {C.}~\bibnamefont
  {Eckart}}\ and\ \bibinfo {author} {\bibfnamefont {G.}~\bibnamefont {Young}},\
  }\bibfield  {title} {\enquote {\bibinfo {title} {The approximation of one
  matrix by another of lower rank},}\ }\href@noop {} {\bibfield  {journal}
  {\bibinfo  {journal} {Psychometrika}\ }\textbf {\bibinfo {volume} {1}},\
  \bibinfo {pages} {211} (\bibinfo {year} {1936})}\BibitemShut {NoStop}%
\end{thebibliography}%
\end{document}